\documentclass[letterpaper, 10 pt, conference]{ieeeconf}
\usepackage[utf8]{inputenc} 
\IEEEoverridecommandlockouts  
\usepackage{amsmath,amsfonts}
\usepackage{amssymb}
\usepackage{amsthm}
\usepackage{array}
\usepackage{graphicx}
\usepackage{caption}

\usepackage{textcomp}
\usepackage{subfigure}
\usepackage{pifont}
\usepackage[justification=justified, font=small]{caption}
\usepackage{stfloats}
\usepackage{color}
\usepackage{multirow}
\usepackage{url}
\usepackage{verbatim}
\usepackage{tikz}
\usepackage{xcolor}
\usepackage{mdframed}
\usepackage{hyperref}
\let\labelindent\relax
\usepackage{enumitem}
\usepackage{balance}
\usepackage{cite}
\usepackage[linesnumbered,ruled,vlined]{algorithm2e}
\usepackage{booktabs}
\usepackage[normalem]{ulem}
\usepackage{float}
\usepackage{courier}

\usepackage[most]{tcolorbox}
\usepackage{graphicx}
\usepackage{tcolorbox}
\tcbset{
	runningexample/.style={
		colback=white,
		colframe=black,
		width=\linewidth,
		boxrule=0.6pt,
		arc=2pt,
		auto outer arc,
		breakable
	}
}
\usetikzlibrary{arrows.meta, positioning, fit, calc, bending, patterns.meta}
\tikzset{
	circ/.style={
		circle, draw=black, fill=white,
		inner sep=1pt, minimum size=18pt, font=\scriptsize
	},
	infobox/.style={
		draw=black, dotted, rounded corners=3pt, line width=0.5pt,
		fill=none, inner sep=5pt
	},
	arr/.style={-{Stealth[length=3pt, width=2.5pt]}, line width=0.5pt},
	lbl/.style={font=\scriptsize},
	ptitle/.style={font=\scriptsize\bfseries},
}

\newcommand{\Apol}{\sigma^{\text{act}}}

\newcommand{\Opol}{\sigma^{\text{op}}}

\newtheorem{corollary}{Corollary}
\newtheorem{definition}{Definition}
\newtheorem{theorem}{Theorem}

\newtheorem{remark}{Remark}

\newtheorem{assumption}{Assumption}

\newcommand{\rat}{{\mathsf{rat}}}
\newcommand{\rob}{{\mathsf{rob}}}

\renewcommand{\O}{\mathcal{O}}

\title{Opinion-Guided Layered Strategies for Decentralized Coordination}

\author{Shuhao Qi$^{1}$, Zhiyong Sun$^{2}$, Siep Weiland$^{1}$, and Sofie Haesaert$^{1}$%
\thanks{This work is supported by the Horizon Europe project AIGGREGATE under grant
No.~101202457, the European project COVER under grant No.~101086228, and the
European project SymAware under grant No.~101070802.}%
\thanks{$^{1}$S. Qi, S. Weiland, and S. Haesaert are with the Department of
Electrical Engineering, Eindhoven University of Technology, Eindhoven, The
Netherlands.
{\tt\small \{s.qi, s.weiland, s.haesaert\}@tue.nl}}%
\thanks{$^{2}$Z. Sun is with the College of Engineering, Peking University,
Beijing, China.
{\tt\small zhiyong.sun@pku.edu.cn}}%
}

\begin{document}
\maketitle

\begin{abstract}
Autonomous agents increasingly interact with other independent agents, and such interactions typically admit multiple joint behaviors. When two agents prefer different ones, their independent strategies may be mutually incompatible and fail to reach a coordinated outcome; when they are identical, neither can differentiate its role when needed. Ideally, an agent should coordinate with any agent it encounters, regardless of which admissible joint behavior that agent aims to realize. We therefore propose a new form of strategy, the opinion-guided strategy, which keeps all the admissible joint behaviors available and postpones the selection to execution time, when the other agent's behavior reveals which one to realize. To realize this, nonlinear opinion dynamics are leveraged in a layered realization to guide the agent to a common admissible joint behavior in response to the other agent's evolving behavior, even without communication. We formally establish the conditions under which the strategy remains robust to every preference the other agent may hold. This robustness has an important implication: two agents running identical strategies can break symmetry when needed, a capability that conventional strategies lack. Three case studies across different applications show that the opinion-guided strategy coordinates with every randomly encountered agent, as long as it is willing to realize one of the admissible joint behaviors. One of them corresponds to a general-sum game: unlike conventional approaches devoted to finding a unique Nash equilibrium in advance, the opinion-guided strategy keeps every equilibrium open and guarantees the agents reach one, decided by their runtime interaction.
\end{abstract}
\pagestyle{plain} 

\section{Introduction}

As a growing number of autonomous agents are deployed in the real world, each agent acts to accomplish its own tasks and satisfy its safety requirements, while interacting with other independent agents whose strategies are unknown. However, simultaneously meeting these objectives for all interacting agents remains an open problem. An increasing number of pathological coordination failures between autonomous vehicles (AVs) illustrate this challenge: as recently reported, multiple AVs from the same company have been observed becoming stuck and causing deadlocks in narrow, unsignalized intersections~\cite{YouTubeVideo_jam} and in parking lots~\cite{YouTubeVideo_honk}. To understand the underlying cause, consider the scenario in Fig.~\ref{fig:hutong}, where two vehicles encounter one another in a narrow corridor and must swap positions to proceed. Such situations are common in constrained environments such as dense residential areas. Intuitively, two joint behaviors can resolve the deadlock: (i) one AV pulls into the waiting lane while the other proceeds first, and (ii) the roles are reversed. Although the two joint behaviors are symmetric, each requires the AVs to play different roles. If the AVs agree on one of them, the conflict is resolved quickly; in practice, however, different vehicles may independently commit to incompatible choices. Worse still, the reports in~\cite{YouTubeVideo_jam, YouTubeVideo_honk} note that the participating AVs were from the same company and thus equipped with identical strategies, leaving them structurally unable to differentiate their roles in such symmetric situations.

We refer to the above as the \emph{decentralized coordination challenge}: given a set of admissible joint behaviors, interacting agents must realize one of them without knowing the choice of the other agent. The central problem is that, even when every agent is rational and committed to realizing an admissible joint behavior, their independent strategies may be mutually incompatible. This incompatibility takes two forms: (i) different agents may disagree on which joint behavior is preferable; and (ii) identical agents necessarily select the same action in symmetric situations, which cannot realize a joint behavior that requires distinct roles. In practice, an autonomous agent may encounter either kind of opponent, and cannot know in advance which. Ideally, an agent should be able to coordinate with any agent it encounters, as long as that agent is likewise rational and committed to one of the admissible joint behaviors. 


\begin{figure}[t]
    \centering
    \includegraphics[width=0.95\linewidth]{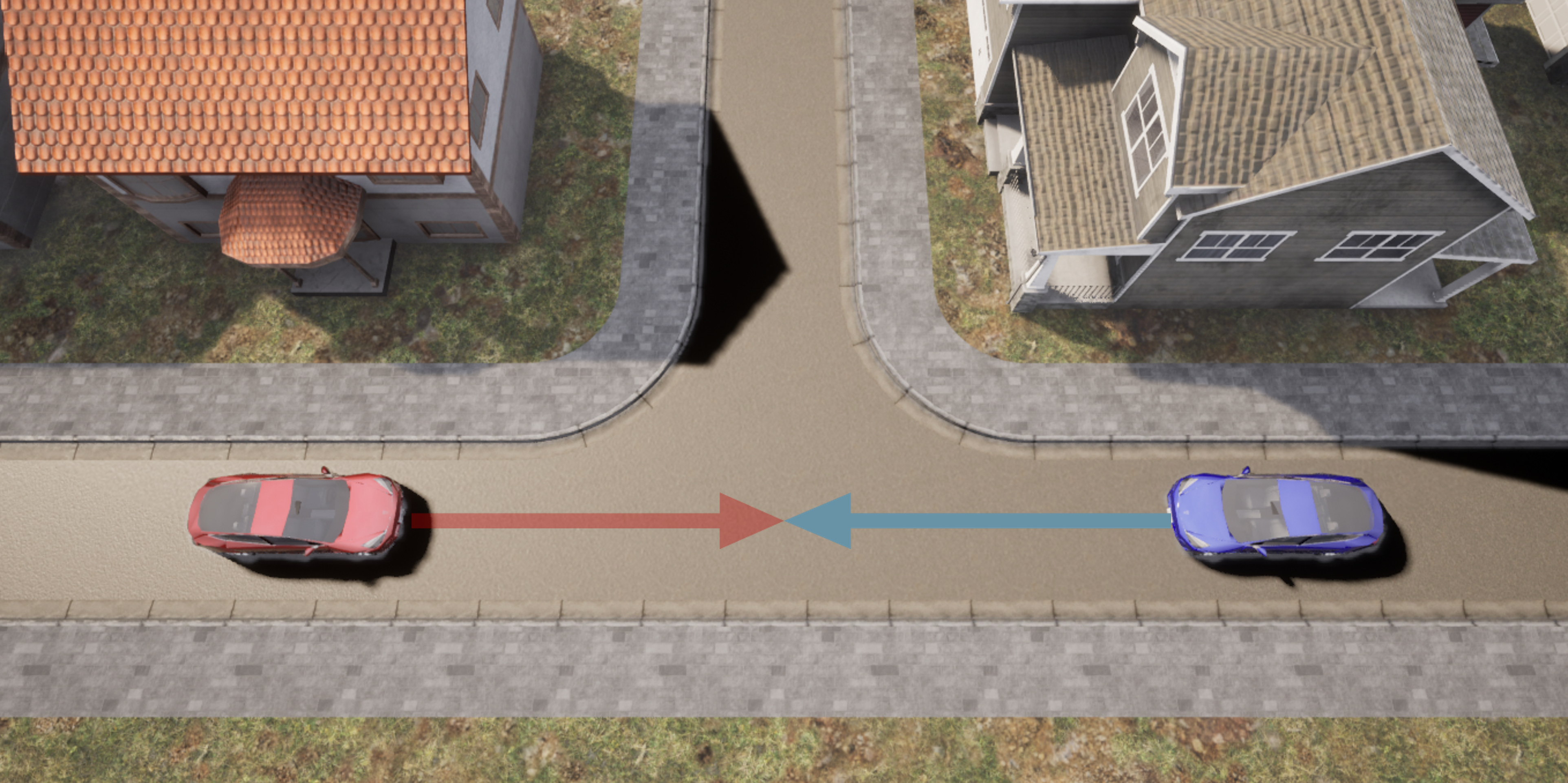}
    \caption{A typical encounter scenario in dense residential regions.}
    \label{fig:hutong}
\end{figure}

\subsection{Related Works}
A well-known instance of this decentralized coordination challenge is the general-sum game, which may admit multiple equilibria~\cite{bernet2002permissive, brice2025permissive}. This multiplicity has prompted much discussion in game theory and multi-agent reinforcement learning (MARL) over how agents should select among coexisting equilibria. When multiple equilibria coexist, decentralized best-response and learning dynamics may fail to converge, cycling among strategy profiles rather than settling on a single equilibrium~\cite{zhang2021multi, zinkevich2005cyclic}. To address this, classical game theory introduces selection criteria such as risk dominance and payoff dominance~\cite{harsanyi1988general}, along with equilibrium refinements
such as trembling-hand perfection~\cite{selten1975reexamination, farina2018practical}, while MARL research designs learning dynamics that steer convergence toward a desired equilibrium~\cite{zhang2024equilibrium, christianos2023pareto}.

In more practical settings, both the game-theoretic and MARL approaches sidestep this challenge and try to select a unique equilibrium in advance by injecting asymmetry, either \emph{explicitly}, through predefined roles and priorities as in Stackelberg games~\cite{zhang2020bi, fisac2019hierarchical}, level-$k$ reasoning~\cite{li2018game}, role-oriented learning~\cite{wangrode}, priority scheduling~\cite{ingebrand2023decentralized}, and ID-conditioned policies~\cite{fu2022revisiting}; or \emph{implicitly}, through jointly trained policies that fix coordination during training. Specifically, the commonly used centralized-training-with-decentralized-execution (CTDE)
paradigm~\cite{oliehoek2016concise, ozdaglar2023independent} settles on a joint behavior at training, which each agent then follows independently at execution. In summary, these explicit and implicit approaches fix the joint outcome at design or training time. Since agents encounter one another at random in practice, as autonomous vehicles do, such predefined asymmetry lacks the flexibility to handle an opponent produced by a different design or training procedure. In addition, the approach of injecting asymmetry overlooks the situation in which the encountered agent runs the identical strategy, where both agents would, for instance, claim priority over the other. 


As in game-theoretic and MARL research, most prior work in multi-agent planning and control also sidesteps the decentralized coordination challenge mentioned above. The existing conflict-resolution methods can be broadly grouped into three categories. (i)~\emph{Fixed rules and priorities}, such as right-hand-priority schemes~\cite{pierson2020weighted}, can be encoded directly into the controller, but they fail when another agent does not follow the same rules and priorities. (ii)~\emph{Perturbation-based methods} inject disturbances to break strict conflict conditions such as deadlock~\cite{wang2017safety}; although practical, they do not guarantee resolution and may violate safety constraints. (iii)~\emph{Communication-based coordination} resolves deadlocks through explicit information exchange. \emph{Centralized} schemes employ a coordinator that jointly manages agent interactions~\cite{mo16cdc}; while effective in small systems, they scale poorly as the number of agents grows~\cite{tomlin98TAC, pritchett2017negotiated}. By contrast, \emph{distributed} schemes rely on local communication among neighbors to negotiate passage, using techniques such as parametric control Lyapunov functions~\cite{weng2022convergence} and rotation controllers for position swaps~\cite{grover2023before, Arul2021iros}. In both cases, their reliance on communication makes them vulnerable to time delays, protocol ambiguity, and disturbances, particularly across heterogeneous agents~\cite{master2020}.

In contrast, human drivers rarely fix on a single plan ahead of time; instead, they keep several possible behaviors open and settle on one only as cues from the ongoing interaction make the right choice clear. This adaptive mechanism is precisely what existing autonomous systems are missing. Drawing on insights from social and biological sciences, the nonlinear opinion dynamics (NOD) framework can model interactions among agents~\cite{bizyaeva2022nonlinear, leonard2024fast}: each agent maintains a continuous opinion state representing its degree of agreement with a decision, updated in response to others' behavior. Moreover, NOD undergoes a supercritical pitchfork bifurcation at a critical attention level, and this bifurcation property enables agents to break symmetry and rapidly converge to a collaborative decision~\cite{bizyaeva2022nonlinear, leonard2024fast}. Crucially, this property targets the coordination challenge identified above. NOD has accordingly been used across multi-agent interaction tasks, such as resolving deadlocks~\cite{cathcart2023proactive}, generating adaptive priorities in narrow corridors~\cite{alghamdi2025opinion}, and stabilizing repeated and differential games~\cite{park2021tuning, hu2023emergent}. In our prior work, we integrated NOD into a safety filter to resolve blocking behaviors in a two-airplane encounter~\cite{qi2025opinion}, without relying on communication or fixed rules. However, NOD is continuous in nature, so existing works have primarily deployed it within a continuous-time controller over a continuous state space. In addition, these integrations are hand-tailored by experts for each specific task. Although the bifurcation property of NOD holds promise for the decentralized coordination challenge, coupling it to a continuous controller in this task-specific manner limits its scalability to general decision-making problems of multi-agent coordination.

\subsection{Contributions}
To address the decentralized coordination challenge, we propose a new form of strategy that lifts the core properties of NOD to the decision-making level. However, such lifting is non-trivial: it requires the resulting framework to preserve the bifurcation and symmetry-breaking properties of NOD while also overcoming the decentralized coordination challenge. In this work, our contributions are fourfold. First, we formalize the decentralized two-agent coordination problem by defining a set-valued \emph{permissive joint strategy} and a \emph{rational individual strategy}, together with the two requirements a strategy must satisfy under independent execution, \emph{robust coordination} and \emph{homogeneous implementability} (Section~\ref{sec:problem}). Second, we propose an \emph{opinion-guided strategy} and analyze the conditions under which it guarantees robust coordination against different types of rational opponents, and homogeneous implementability, even in symmetric situations (Sections~\ref{sec:opinion_guided} and~\ref{sec:analysis}). Third, we propose a layered framework that realizes the opinion-guided strategy, comprising opinion dynamics, a continuous controller, and an opinion estimator, which together enable each agent to adapt in real time to the evolving behavior of others and converge to a common admissible joint behavior while operating fully decentralized, without inter-agent communication (Section~\ref{subsec:comm_free}). Fourth, we evaluate the proposed framework across three case studies, the last of which corresponds to a classic general-sum game (Section~\ref{sec:case_study}).

\section{Background: The Coordination Challenge}

\subsection{Running Example: Charging-Bay Selection}\label{sec:example}

\begin{figure}[thb]
    \centering
    \includegraphics[width=0.55\linewidth]{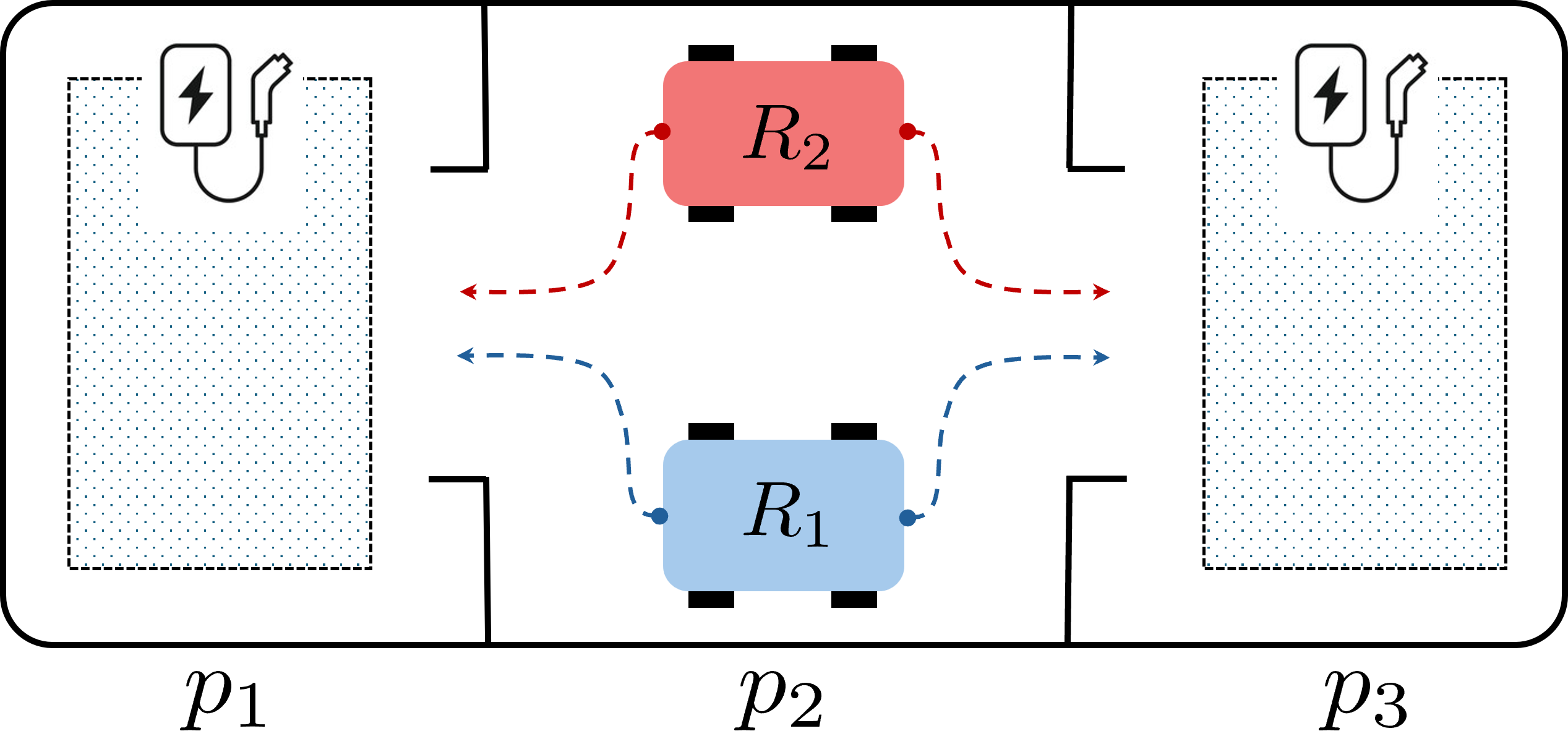}
    \caption{Sketch of the running example, a two-robot charging-bay selection task that instantiates the anti-coordination game. The robots $R_1$ (blue) and $R_2$ (red) start together in the middle zone $p_2$, while a charging bay is available in each of the side zones $p_1$ and $p_3$. The dashed arrows mark the two candidate moves of each robot (left toward $p_1$ or right toward $p_3$), and the task succeeds only when the two robots commit to distinct bays, so that both can recharge simultaneously.}
    \label{fig:example}
\end{figure}

To make the decentralized coordination challenges concrete, we begin with a simple two-robot example that instantiates the classical \emph{anti-coordination game}, of which the ``chicken game'' is a well-known instance~\cite{bramoulle2007anti}. In such a game, the agents are expected to take distinct actions to reach a desired outcome. 

For simplicity, we consider a single-stage game, so that each agent makes only one decision. As shown in Fig.~\ref{fig:example}, the two robots, denoted by $R_1$ and $R_2$, start at the middle zone $p_2$ and must be dispatched simultaneously to the two charging bays located in the target zones $p_1$ and $p_3$, with one robot assigned to each. Since each bay can serve only one robot at a time, and each robot selects its target independently and without knowledge of the other's choice, the task succeeds only when the two robots commit to different zones. Two admissible joint behaviors thus arise: $R_1$ moves to $p_1$ while $R_2$ moves to $p_3$, or vice versa. Since $R_1$ and $R_2$ may in reality be either heterogeneous or homogeneous, this minimal setting exposes two failure modes of decentralized coordination:
\begin{enumerate}[leftmargin=*, align=left]
    \item \textbf{Miscoordination risk:} The robots complete the charging task only when they independently select compatible strategies ($R_1$ to $p_1$, $R_2$ to $p_3$), yet they may just as well select incompatible ones (both to $p_1$ or both to $p_3$), and such compatibility can never be ensured.
    \item \textbf{Symmetry-breaking failure:} If the robots are homogeneous, whatever strategy they adopt, they make identical decisions concurrently in this symmetric scenario (both to $p_1$ or both to $p_3$) and thus can never commit to the distinct zones the task requires.
\end{enumerate}

\subsection{A Promising Approach: Nonlinear Opinion Dynamics}\label{sec:nod}

\begin{figure}[htb]
   \centering 
   \subfigure[Clockwise swap]{ \includegraphics[trim=100 0 100 0, clip, width=0.465\linewidth]{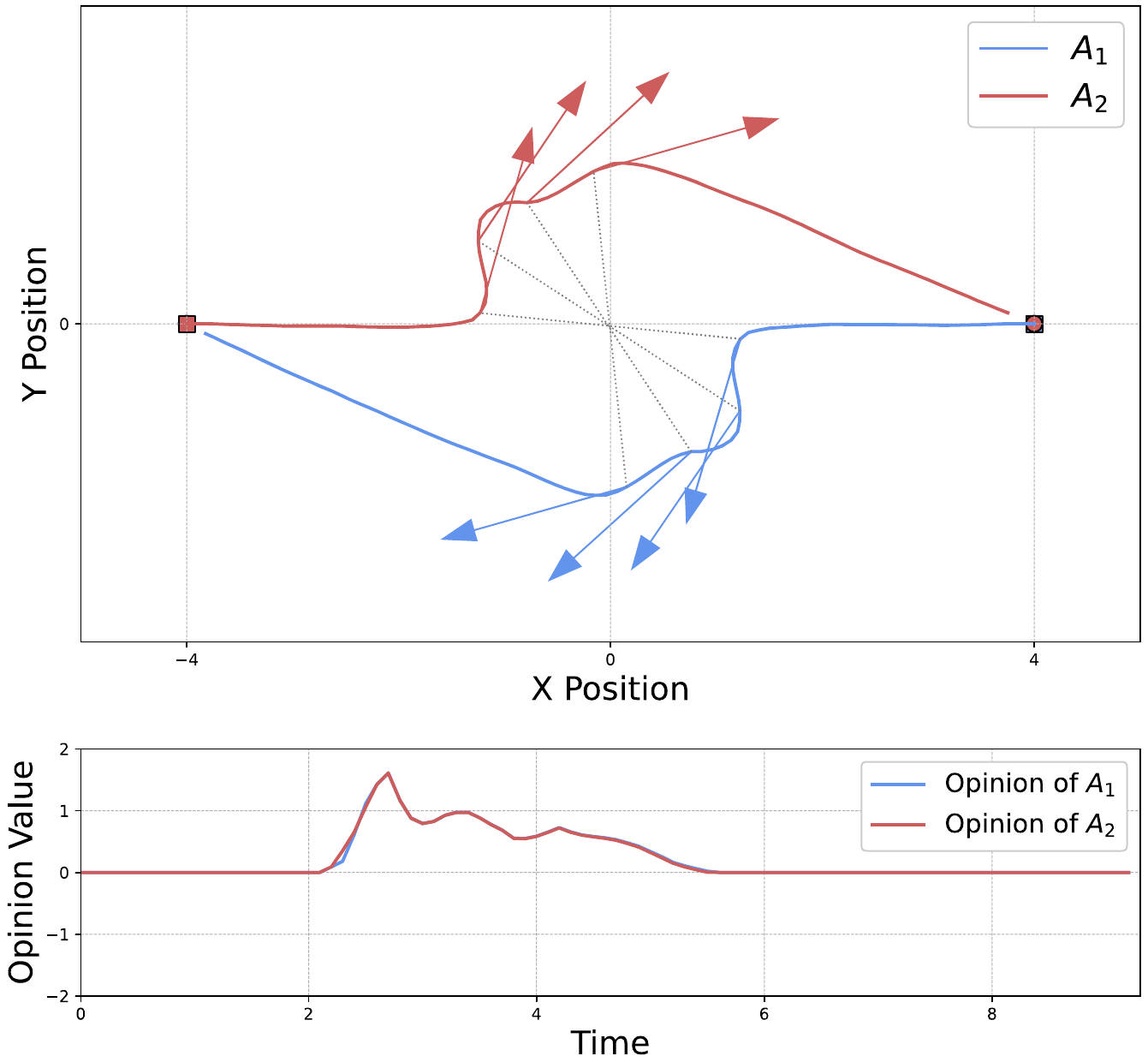} }
    \subfigure[Anti-clockwise swap]{ \includegraphics[trim=100 0 100 0, clip, width=0.465\linewidth]{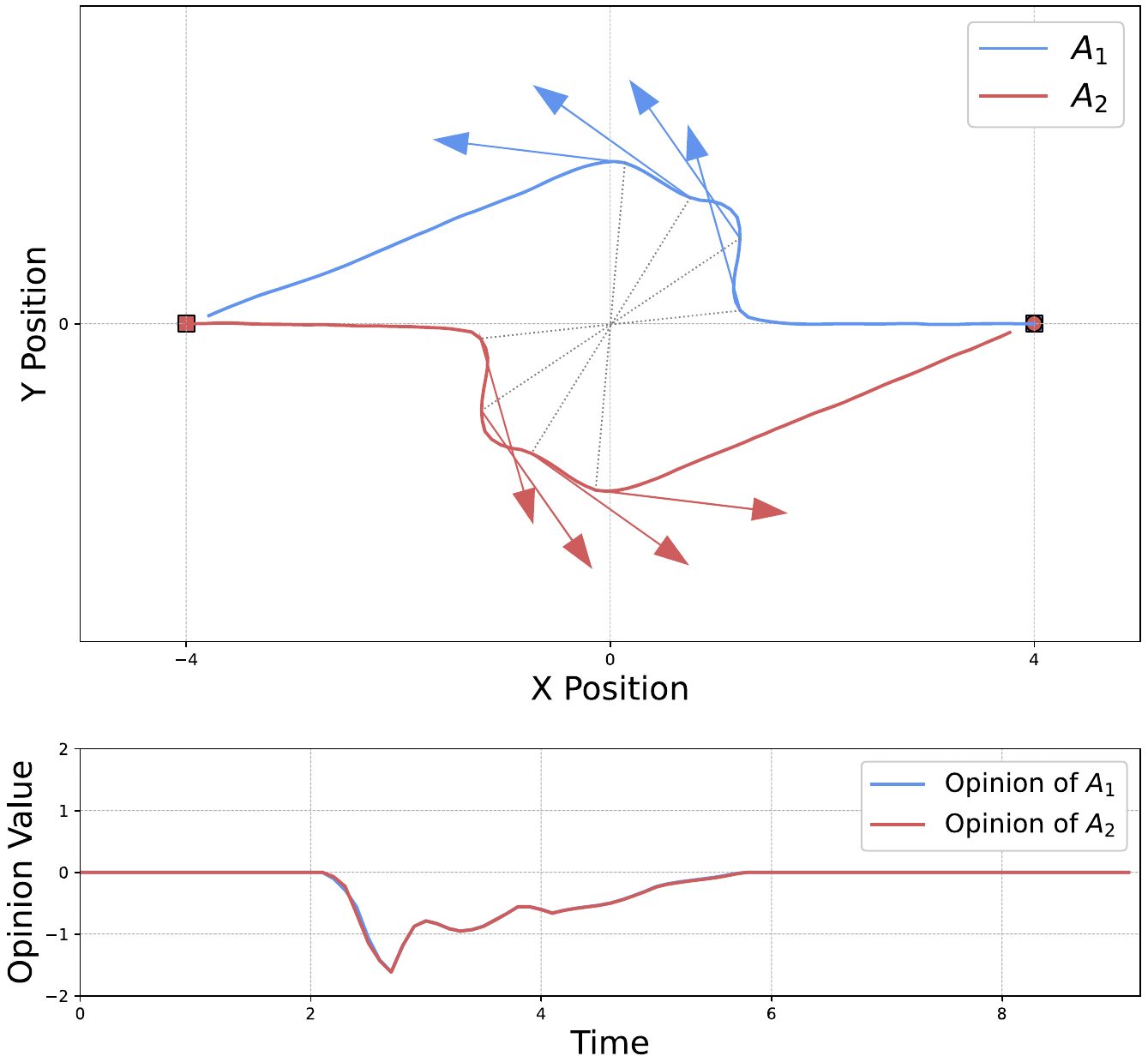} }
    \caption{Two resolution behaviors randomly generated in a two-airplane encounter. The upper panels show the airplanes' trajectories, while the bottom panels depict the evolution of the opinion states.
    }
    \label{fig:airplane_sim}
\end{figure}
In our prior work~\cite{qi2025opinion}, we proposed an approach based on NOD to mitigate the two coordination challenges described above at the continuous-control layer. This approach does not require predefined rules, fixed priorities, or inter-agent communication. Intuitively, an agent's opinion state represents its degree of agreement with a decision. The bifurcation property of NOD~\cite{leonard2024fast} makes each agent highly sensitive to small changes in the other's state, enabling the group to break symmetry when needed and rapidly converge to a collaborative decision. As shown in Fig.~\ref{fig:airplane_sim}, two airplanes are tasked to swap their positions, for which two symmetric joint behaviors are admissible: bypassing on the left or bypassing on the right. Each airplane infers the other's opinion state, which encodes its intent of the bypassing side, and updates its own opinion based on local observations. As a result, if one airplane commits to one bypassing side, the other can adapt to it and reach consensus on that side. If both airplanes are free to choose, the final outcome has two possibilities: clockwise or anticlockwise swapping (Fig.~\ref{fig:airplane_sim}). Which one is realized is determined by the real-time evolution of the opinion states. Therefore, the NOD-based controller can resolve both miscoordination and symmetry-breaking failures in a fully decentralized manner. This paper builds on this insight by lifting the core properties of NOD from the continuous-control layer to the discrete decision-making layer, enabling more general strategic coordination.

\section{Problem Formulation and Statement}\label{sec:problem}
This section formalizes the decentralized coordination problem. To examine this challenge in depth, we focus on two-agent interactions, an important step before considering interactions among more agents.

\subsection{Game Model and Strategies}

\begin{definition}[Two-Agent Game]\label{def:game}
A two-agent game is defined as a tuple
$
\mathcal{G} = (\mathcal{S}, s_0, A_1, A_2, \rightarrow),
$
where:
\begin{itemize}
    \item $\mathcal{S}$ is the set of joint states, with initial state $s_0 \in \mathcal{S}$;
    \item $A_1$ is the set of actions of Agent~1;
    \item $A_2$ is the set of actions of Agent~2; and
    \item $\rightarrow \subseteq \mathcal{S} \!\times\! (A_1 \!\times\! A_2) \!\times\! \mathcal{S}$
    is the transition relation that maps joint state and joint action to a successor state.
\end{itemize}
\end{definition}
\noindent An infinite sequence of joint actions $(a_{1,0}, a_{2,0}), $ $(a_{1,1}, a_{2,1}), (a_{1,2}, a_{2,2}), \ldots$ induces an execution of $\mathcal{G}$:
$$
s_0 \xrightarrow{(a_{1,0}, a_{2,0})} s_1 
\xrightarrow{(a_{1,1}, a_{2,1})} s_2 
\xrightarrow{(a_{1,2}, a_{2,2})} s_3 \ldots
$$
where $(s_t, (a_{1,t}, a_{2,t}), s_{t+1}) \in \rightarrow$ for all $t \ge 0$. We define the set of all possible finite executions (i.e., histories), based on which decisions can be made, as
$\mathcal{H} := (\mathcal{S} \times A_1 \times A_2)^\ast \times \mathcal{S}$,
where $X^\ast$ denotes the Kleene star, i.e., the set of all finite sequences over a set $X$. An element $h \!\in \mathcal{H}$ is a finite sequence of the form
$h = s_0 \xrightarrow{(a_{1,0}, a_{2,0})} s_1
\xrightarrow{(a_{1,1}, a_{2,1})} \cdots
\xrightarrow{(a_{1,t-1}, a_{2,t-1})} s_t$.  

\noindent\emph{Notation.} The labels $1,2$ denote the two specific agents. When describing the interaction between them, we use $i \in \{1,2\}$ to denote a generic agent and $j \in \{1,2\}\setminus\{i\}$ the other agent, so that $\{i,j\}=\{1,2\}$.

 \noindent\textbf{Individual strategy.} For each agent $i \!\in\! \{1, 2\}$, a deterministic individual strategy returns an individual action $a_i \!\in\! A_i$, denoted by $\sigma_i$. A commonly used individual strategy is a mapping from the state history to an action~\cite{hespanha2017noncooperative,kochenderfer2015decision}, i.e., $\sigma_i \!:\! \mathcal{H} \!\rightarrow\! A_i$. The strategy is \emph{memoryless (Markov)} if it is a function $\sigma_i : \mathcal{S} \rightarrow A_i$, meaning that the chosen action depends only on the current state. Given two individual strategies $\sigma_1 : \mathcal{H} \rightarrow A_1$ and $\sigma_2 : \mathcal{H} \rightarrow A_2$, the \emph{joint strategy} is the function $\sigma_1 \,\|\, \sigma_2 : \mathcal{H} \rightarrow A_1 \times A_2$, defined for each history $h \in \mathcal{H}$ by
\begin{equation}
	(\sigma_1 \,\|\, \sigma_2)(h) := \{(\sigma_1(h), \sigma_2(h))\}.
\end{equation}

\noindent\textbf{Permissive joint strategy.} 
A strategy for a multi-agent system is conventionally defined as a single-valued mapping $\mathcal{H} \!\rightarrow\!A_1 \!\times\!A_2$, which is suitable for a centralized approach in which all agents are coordinated to follow one prescribed joint action at each history. As illustrated in the running example, however, a multi-agent decentralized system may admit multiple equally admissible joint actions. Prior work has captured this multiplicity through set-valued strategies, such as permissive strategies in reactive synthesis~\cite{demri2022ltlf}, permissive controllers for interval MDPs~\cite{badings2025robust}, and permissive equilibria in multi-player reachability games~\cite{brice2025permissive}. Following these definitions, we define a \emph{permissive joint strategy} as a set-valued mapping from each history to a set of joint actions, given by
$$
\pi: \mathcal{H} \rightarrow 2^{A_1 \times A_2},
$$
where $\pi(h) \subseteq  A_1 \times A_2$ represents the set of joint actions permitted at history $h \in \mathcal{H}$. Throughout the paper, a joint action in $\pi(h)$, and the joint behavior it induces, are called \emph{admissible}.

\noindent\textbf{Rational actions.}
In practice, the challenge is that each agent should select an action without knowledge of the other agent's choice, such that the resulting joint behavior may fail to be admissible, $(\sigma_1 \,\|\, \sigma_2)(h) \not\subseteq \pi(h)$. In the extreme case where an agent behaves adversarially, coordination cannot be achieved. However, in realistic applications such as traffic scenarios, adversarial behavior is rare, as each vehicle pursues its own navigation objective rather than seeking to attack others. The literature on \emph{rational 
synthesis}~\cite{fisman2010rational,kupferman2016synthesis} has shown that specifications unrealizable against a hostile environment can become realizable once agent rationality is taken into account. Given this motivation, we assume each agent is rational and committed to realizing an admissible joint behavior within the shared permissive strategy $\pi$, which is referred to as \emph{$\pi$-rationality}. The rational action set and individual strategy under $\pi$-rationality are formalized below.

\begin{definition}[$\pi$-Rational Action Set and Strategy]~\label{def:rational}
	Let $\pi$ be a given permissive joint strategy for a two-agent game $\mathcal{G}$ defined in Def.~\ref{def:game}. Given a history $h$ and $i\in\{1,2\}$, an action $a_i \!\in\! A_i$ of agent $i$ is rational with respect to $\pi$ at history $h$ if there exists $a_{j} \in A_{j}$ such that $(a_1, a_{2}) \!\in\! \pi(h)$. The rational action set of agent $i$ at history $h$ with $\pi(h)\!\neq\! \emptyset$ is defined as
	\begin{equation}
		A^\rat_i(h) = \{a_i \in A_i \mid \exists\, a_{j} \in A_{j} : (a_1, a_{2}) \in \pi(h)\}.
	\end{equation}
	An individual strategy $\sigma_i$ is said to be \emph{rational} with respect to $\pi$ if $\sigma_i(h) \in A^\rat_i(h)$ for all $h \!\in\! \mathcal{H}$ with $\pi(h)\!\neq\! \emptyset$.
\end{definition}
Intuitively, a rational individual strategy is one that realizes an admissible joint behavior under $\pi$, provided the other agent behaves cooperatively. We denote the set of all the rational strategies for agent $i$ by $\Sigma_i^\rat$, where $i \in \{1, 2\}$. We use the running example to illustrate these concepts. 

\begin{tcolorbox}[runningexample]

\textbf{Running example}: In the two-robot anti-coordination game, the state of each robot is its location, drawn from $\{p_1, p_2, p_3\}$, where $p_2$ is the initial zone and $p_1$, $p_3$ are the two target zones. The joint state set $\mathcal{S}$ consists of all possible combinations of the robots' positions, with initial state $s_0\!=\!(p_2, p_2)$. Each robot $R_i$ selects an action from $A_i\!=\!\{a_{p_1}, a_{p_2}, a_{p_3}\}$, where $a_{p_k}$ denotes movement toward zone $p_k$ for $k\!\in\!\{1,2,3\}$; in particular, $a_{p_1}$ and $a_{p_3}$ move the robot toward the target zones $p_1$ and $p_3$ respectively, and $a_{p_2}$ keeps the robot stationary at $p_2$. The objective is to reach a terminal state $(p_1, p_3)$ or $(p_3, p_1)$. A permissive joint strategy at the initial state is illustrated in Fig.~\ref{fig:example_strategy}, which includes two valid cooperative action pairs: $\pi(s_0) = \{(a_{p_1}, a_{p_3}),\ (a_{p_3}, a_{p_1})\}$. The corresponding rational action set for each robot is $A^\rat_i(s_0)\!=\!\{a_{p_1}, a_{p_3}\}$ for $i\!\in\!\{1,2\}$.
\begin{center}
	\includegraphics[width=0.62\linewidth]{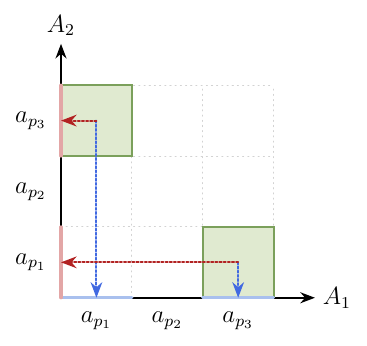}
\end{center}
\captionof{figure}{A permissive joint strategy for the running example at the initial state. The green regions represent the permissive joint action pairs $\pi(s_0)$. The dashed arrows show projections to the rational action sets.}
\label{fig:example_strategy}
\end{tcolorbox}


\subsection{Robust Coordination}

\begin{figure*}[hb] 
	\centering
	\includegraphics[width=\textwidth]{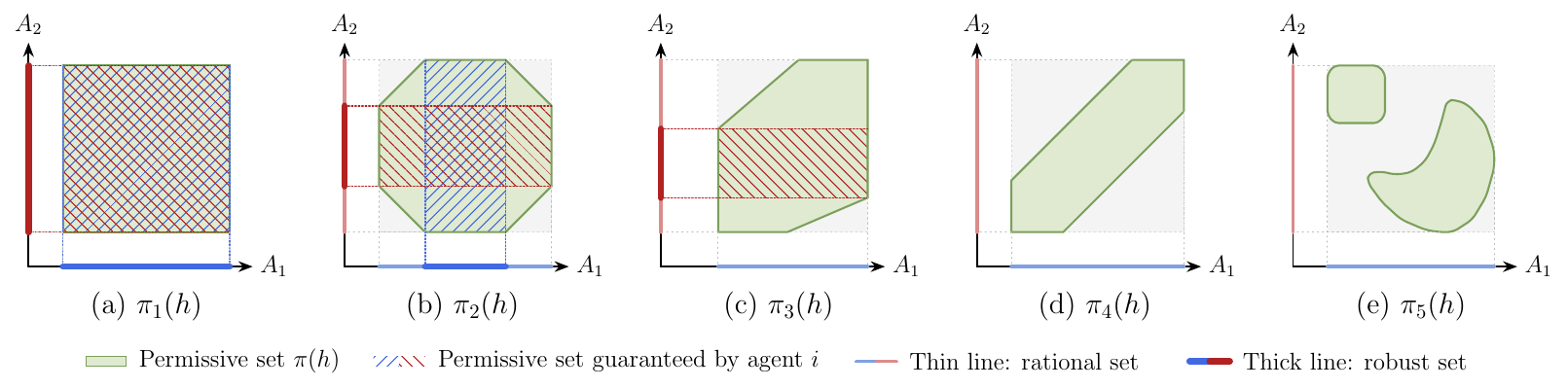}
	\caption{Representative geometries of a permissive joint strategy, $\pi_1(h)$ through $\pi_5(h)$, illustrating cases in which both, one, or neither agent has a non-empty robust set. In every panel the horizontal axis is $A_1$ (Agent~1, blue) and the vertical axis $A_2$ (Agent~2, red). On the axes, thin light lines mark the rational sets $A_i^\rat(h)$ and thick dark lines the robust sets $A_i^\rob(h;\Sigma_j^\rat)$, and the dotted gray box is their product $A_1^\rat(h)\!\times\!A_2^\rat(h)$. The green region is the permissive set $\pi(h)$; blue and red hatching mark the joint actions in which Agent~1 and Agent~2, respectively, commit to their robust set.}
	\label{fig:coordination_analysis}
\end{figure*}

When an agent faces another whose strategy is unknown, an ideal strategy is one that guarantees coordination regardless of the other agent's choice. We refer to such a strategy as a \emph{robust strategy}, which is formally defined below.

\begin{definition}[$(\pi,\Sigma_{j})$-Robust Strategies]~\label{def:robust_coord}
	Given a permissive strategy $\pi$ and a set of strategies $\Sigma_{j}$ of agent $j$, we say that an individual strategy $\sigma_i$ of agent $i$ is  $(\pi,\Sigma_{j})$-robust if for all $h\!\in\!\mathcal{H}$ with $\pi(h)\!\neq\!\emptyset$:
	$$\forall \sigma_j\!\in\!\Sigma_{j}:\; (a_1,a_2)\!\in\!\pi(h),$$ where $a_i\!=\!\sigma_i(h),\, a_j\!=\!\sigma_j(h)$, $i,j\!\in\!\{1,2\}$ and $i\!\neq\!j$.
\end{definition}

At each history $h$, the actions that agent $i$ can select to guarantee coordination regardless of the other agent's choice are called \emph{robust actions}, which are defined as follows.

\begin{definition}[$(\pi,\Sigma_{j})$-Robust Action Set]\label{def:robust_set}
Given a permissive strategy $\pi$ and a set of strategies $\Sigma_{j}$ of agent $j$, let $A^{\Sigma_{j}}(h)\!:=\!\{\sigma_j(h)\mid\sigma_j\!\in\!\Sigma_{j}\}$ denote the set of actions agent $j$ may select at history $h$ under a strategy in $\Sigma_{j}$. The \emph{robust action set} of agent $i$ at a history $h$ with $\pi(h)\!\neq\!\emptyset$, denoted by $A^\rob_i(h; \Sigma_{j})$, is defined as
\begin{equation}
\begin{aligned}
	A^\rob_i(h; \Sigma_{j}) \!:=\! \big\{ a_i\!\in\!A_i \;\big|\; & \forall a_j\!\in\! A^{\Sigma_{j}}(h): \\ & (a_1,a_2)\!\in\!\pi(h) \big\},
\end{aligned}
\end{equation}
where $i, j\!\in\!\{1,2\}$ and $i\!\neq\!j$.
\end{definition}

 If agent $i$ selects an action from a nonempty robust action set $A^\rob_i(h;\Sigma_{j})$ at every history $h$, then coordination is guaranteed regardless of which strategy $\sigma_j \!\in\! \Sigma_{j}$ agent $j$ selects, providing a unilateral guarantee of coordination. In particular, a $(\pi,\Sigma_j^{\rat})$-robust strategy guarantees the coordination prescribed by $\pi$ against every rational strategy of the other agent. When $\Sigma_j \!=\! \Sigma_j^\rat$, we abbreviate $A_i^\rob(h)\!:=\!A_i^\rob(h; \Sigma_j^\rat)$.


Figure~\ref{fig:coordination_analysis} illustrates five representative geometries of the permissive action set $\pi(h)$. Across all subfigures, the light-colored lines along the axes mark the rational action sets $A_i^\rat(h)$ from Def.~\ref{def:rational}, and the dark-colored lines mark the robust action sets $A_i^\rob(h)$ from Def.~\ref{def:robust_set}, $i\in\{1,2\}$. In Fig.~\ref{fig:coordination_analysis}(a), every pair of rational actions is permitted, so $\pi(h)\!=\!A_1^\rat(h)\!\times\!A_2^\rat(h)$ and $A_i^\rat(h)\!=\!A_i^\rob(h)$ for $i\!\in\!\{1,2\}$; coordination is guaranteed regardless of how the agents act.
Fig.~\ref{fig:coordination_analysis}(b) shows a case in which not all pairs of $A_1^\rat(h)\!\times\!A_2^\rat(h)$ lie in $\pi(h)$, yet both robust action sets remain non-empty, so either agent can unilaterally guarantee coordination by restricting its choice to $A_i^\rob(h)$.
In Fig.~\ref{fig:coordination_analysis}(c), only $A_2^\rob(h)$ is non-empty; agent~1 cannot unilaterally guarantee coordination, because no single $a_1\!\in\!A_1^\rat(h)$ satisfies $(a_1,a_2)\!\in\!\pi(h)$ for every $a_2\!\in\!A_2^\rat(h)$. In reality, the permissive action set may have more complex geometries, and it is possible that both robust action sets are empty, as illustrated in Fig.~\ref{fig:coordination_analysis}(d) and~(e). In such cases, neither agent can unilaterally guarantee coordination, even if both agents are rational, and the two agents face the risk of miscoordination: each may select a rational action, yet their joint action can still fall outside the permissive set. It is therefore essential to mitigate this risk by ensuring the existence of robust actions.

\subsection{Homogeneous Implementability}

Beyond the miscoordination risk in the heterogeneous settings, another challenge arises when the two agents are homogeneous~\cite{zinkevich2001symmetry}, as exemplified by the running example. An ideal strategy should therefore not only be robust to the other rational agent, but also remain effective in homogeneous settings, a property we refer to as \emph{homogeneous implementability}.

Building on the notions of homogeneous multi-agent Markov decision processes introduced in~\cite{zinkevich2001symmetry}, we adapt the definition to our two-agent setting. The underlying premise is that the labels ``agent~1'' and ``agent~2'' carry no physical meaning, so every property of the game should be invariant under permuting them. To formalize this, writing the joint state as $\mathcal{S}\!=\!S_1\!\times\! S_2$, we introduce a swap operator $\tau \!: S_1 \!\times\! S_2 \!\to\! S_1 \!\times\! S_2$ as $\tau(s_1, s_2) \!=\! (s_2, s_1).$

\begin{definition}[Homogeneous Two-Agent Game]\label{def:symmetric}
	A two-agent game $\mathcal{G} = (\mathcal{S}, s_0, A_1, A_2, \rightarrow)$ as defined in Def.~\ref{def:game} with $\mathcal{S} = S_1 \times S_2$,
	 $S_1 = S_2$ and $A_1 = A_2$ is called \emph{homogeneous} if for all $s, s' \!\in\! \mathcal{S}$ and $a_1, a_2 \!\in\! A_1$,
		$(s,(a_1,a_2),s') \in\;\rightarrow$ if and only if
		$(\tau(s),(a_2,a_1),\tau(s')) \in\;\rightarrow$.
\end{definition}
\noindent Following~\cite[Def.~14]{zinkevich2001symmetry}, we say that two Markov strategies $\sigma_1$ and $\sigma_2$ are \emph{identical} if the joint action they induce is equivariant under agent relabeling:
$$
(\sigma_1(s), \sigma_2(s)) \;=\; \tau\!\left(\sigma_1(\tau(s)),\, \sigma_2(\tau(s))\right) \quad \forall s \in \mathcal{S}.
$$
This condition is equivalent to the pointwise relation $\sigma_1(s) \!=\! \sigma_2(\tau(s))$ for all $s$. Intuitively, the strategies of two agents are identical when they take the same action from their own first-person view of every state, so that exchanging the labels ``agent~1'' and ``agent~2'' leaves the joint behavior unchanged.

Similarly to $\tau$, we define the permutation of a history,
$h = s_0 \xrightarrow{(a_{1,0}, a_{2,0})} s_1
\xrightarrow{(a_{1,1}, a_{2,1})} \cdots
\xrightarrow{(a_{1,t-1}, a_{2,t-1})} s_t$, as
$\bar{\tau}(h) = \tau(s_0) \xrightarrow{(a_{2,0},a_{1,0})} \tau(s_1)
\xrightarrow{(a_{2,1}, a_{1,1})} \cdots
\xrightarrow{(a_{2,t-1}, a_{1,t-1})} \tau(s_t). $
 This allows us to define identical strategies $\sigma_1$ and $\sigma_2$ as those for which
\begin{equation}\label{eq:homo_identical}
 \sigma_1(h) = \sigma_2(\bar \tau(h)) \quad \forall h\in \mathcal H.
\end{equation}

\begin{definition}[$\pi$-Homogeneous Implementability]\label{def:robust_coord_homo}
Let $\mathcal{G}$ be a homogeneous two-agent game (Def.~\ref{def:symmetric}) and $\pi$ a permissive strategy. An individual strategy $\sigma$ is \emph{$\pi$-homogeneously implementable} if
\begin{equation}\label{eq:homo_implement}
	\forall h \in \mathcal{H} \text{ with } \pi(h)\neq \emptyset: \quad \emptyset\!\neq\! (\sigma \,\|\, \sigma \circ \bar\tau)(h) \subseteq \pi(h),
\end{equation}
where $(\sigma \circ \bar\tau)(h) := \sigma(\bar\tau(h))$. 
\end{definition}

Such identical strategies are known to exist when the two agents are never in the same state~\cite[Theorem~3]{zinkevich2001symmetry}. However, for an individual strategy of the form $\sigma: \mathcal{H} \!\to\! A_1$, the symmetric condition with $\tau(s) \!=\! s$ forces both agents to select the same action, so a $\pi$-homogeneously implementable strategy of this form may not exist.

\begin{tcolorbox}[runningexample]
\textbf{Running example}: In the two-robot example, the initial state $s_0 = (s_{0,1}, s_{0,2})$ satisfies $s_{0,1} \!=\! s_{0,2}$, and $\pi(s_0) = \{(a_{p_1}, a_{p_3}),\ (a_{p_3}, a_{p_1})\}$. Since $\tau(s_0) = s_0$, the homogeneity condition forces both agents to choose the same action, so the feasible joint actions are restricted to the diagonal $\{(a, a) : a \in A_1\}$, none of which lies in $\pi(s_0)$, as shown in Fig.~\ref{fig:example_strategy}.
\end{tcolorbox}

The symmetric situation with $s_1\!=\!s_2$ captures those joint states at which the two agents are indistinguishable. In practice, such symmetry is not rare: two homogeneous agents in a symmetric situation perceive identical first-person observations and therefore act identically, as illustrated by the encounter scenario in Fig.~\ref{fig:hutong}. It is therefore essential that the strategy of an autonomous agent deployed in an interactive environment be homogeneously implementable.

\subsection{Problem Statement}

Consider a two-agent game $\mathcal{G}$ (Def.~\ref{def:game}) with a permissive joint strategy $\pi$ shared by both agents as common knowledge. Each agent selects its action independently. Without loss of generality, take agent~$i$ as the ego agent whose individual strategy $\sigma_i$ is to be designed, and let agent~$j$ be the encountered agent, which follows an unknown $\pi$-rational strategy $\sigma_j\!\in\!\Sigma_j^\rat$ that may either differ from $\sigma_i$ or be identical to it in the sense of~\eqref{eq:homo_identical}. The objective is to design $\sigma_i$ such that, in both cases, the realized joint action lies in $\pi(h)$ at every history $h\!\in\!\mathcal{H}$ with $\pi(h)\!\neq\!\emptyset$. Formally, $\sigma_i$ is required to (i)~achieve $(\pi,\Sigma_{j}^\rat)$-robust coordination (Def.~\ref{def:robust_coord}), and (ii)~be $\pi$-homogeneously implementable (Def.~\ref{def:robust_coord_homo}) whenever $\mathcal{G}$ is homogeneous (Def.~\ref{def:symmetric}).

\section{Opinion-Guided Strategies}\label{sec:opinion_guided}

\subsection{Reactive Strategy}
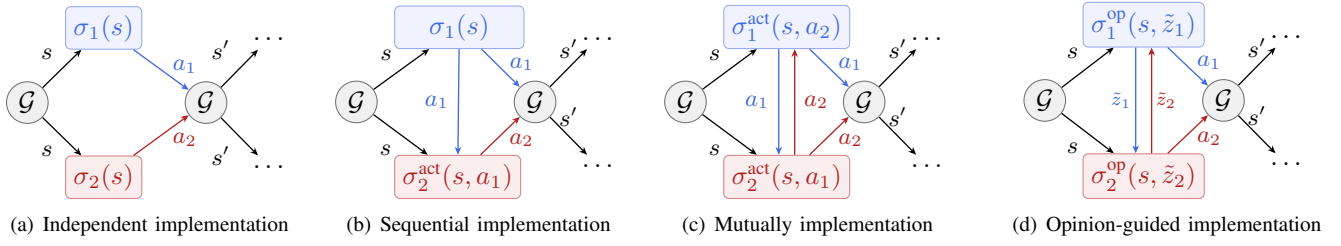
\begin{figure*}[hbt]
	\centering
	\definecolor{agA}{RGB}{65,105,225}
\definecolor{agB}{RGB}{178,34,34}

\tikzset{
  snodeA/.style={text=agA, rounded corners=2pt, draw=agA!55, fill=agA!8, line width=0.4pt, inner sep=3pt},
  snodeB/.style={text=agB, rounded corners=2pt, draw=agB!55, fill=agB!8, line width=0.4pt, inner sep=3pt},
  gstate/.style={circle, draw=black!55, fill=black!6, line width=0.4pt, inner sep=1pt, minimum size=0.55cm},
}

\subfigure[Independent implementation]{%
\begin{tikzpicture}[node distance=1.4cm]
  \node[gstate] (s) {$\mathcal{G}$};
  \node[snodeA, above right=0.5cm and 0.3cm of s] (p1) {$\sigma_1(s)$};
  \node[snodeB, below right=0.5cm and 0.3cm of s] (p2) {$\sigma_2(s)$};
  \node[gstate, right=1.8cm of s]                  (o)  {$\mathcal{G}$};
  \node[rectangle, above right=0.5cm and 0.3cm of o] (p1b) {$\ldots$};
  \node[rectangle, below right=0.5cm and 0.3cm of o] (p2b) {$\ldots$};

  \draw[arr] (s) -- node[above left]{\small$s$} (p1);
  \draw[arr] (s) -- node[below left]{\small$s$} (p2);
  \draw[arr, draw=agA] (p1) -- node[above right, pos=0.65, inner sep=1pt, text=agA]{\small$a_1$} (o);
  \draw[arr, draw=agB] (p2) -- node[below right, pos=0.65, inner sep=1pt, text=agB]{\small$a_2$} (o);
  \draw[arr] (o) -- node[above left]{\small$s'$} (p1b);
  \draw[arr] (o) -- node[below left]{\small$s'$} (p2b);
\end{tikzpicture}}\hfill
\subfigure[Sequential implementation]{%
\begin{tikzpicture}[node distance=1.4cm]
  \node[gstate] (s) {$\mathcal{G}$};
  \node[snodeA, above right=0.5cm and 0.3cm of s] (ldr) {$\quad\sigma_1(s)\quad$};
  \node[snodeB, below right=0.5cm and 0.3cm of s] (flw) {$\Apol_2(s,a_1)$};
  \node[gstate, right=1.8cm of s]                  (o)  {$\mathcal{G}$};
  \node[rectangle, above right=0.5cm and 0.3cm of o] (p1b) {$\ldots$};
  \node[rectangle, below right=0.5cm and 0.3cm of o] (p2b) {$\ldots$};

  \draw[arr] (s) -- node[above left]{\small$s$} (ldr);
  \draw[arr] (s) -- node[below left]{\small$s$} (flw);
  \draw[arr, draw=agA] (ldr) -- node[left, text=agA]{\small$a_1$} (flw);
  \draw[arr, draw=agB] (flw) -- node[below right, pos=0.65, inner sep=1pt, text=agB]{\small$a_2$} (o);
  \draw[arr, draw=agA] (ldr) -- node[above right, pos=0.65, inner sep=1pt, text=agA]{\small$a_1$} (o);
  \draw[arr] (o) -- node[above]{\small$s'$} (p1b);
  \draw[arr] (o) -- node[above]{\small$s'$} (p2b);
\end{tikzpicture}}\hfill
\subfigure[Mutually implementation]{%
\begin{tikzpicture}[node distance=1.4cm]
  \node[gstate] (s) {$\mathcal{G}$};
  \node[snodeA, above right=0.5cm and 0.3cm of s] (ldr) {$\Apol_1(s, a_2)$};
  \node[snodeB, below right=0.5cm and 0.3cm of s] (flw) {$\Apol_2(s,a_1)$};
  \node[gstate, right=1.8cm of s]                  (o)  {$\mathcal{G}$};
  \node[rectangle, above right=0.5cm and 0.3cm of o] (p1b) {$\ldots$};
  \node[rectangle, below right=0.5cm and 0.3cm of o] (p2b) {$\ldots$};

  \draw[arr] (s) -- node[above left]{\small$s$} (ldr);
  \draw[arr] (s) -- node[below left]{\small$s$} (flw);
  \draw[arr, draw=agA, transform canvas={xshift=-3pt}] (ldr) -- node[left, font=\scriptsize, text=agA]{\small$a_1$} (flw);
  \draw[arr, draw=agB, transform canvas={xshift=3pt}]  (flw) -- node[right, font=\scriptsize, text=agB]{\small$a_2$} (ldr);
  \draw[arr, draw=agB] (flw) -- node[below right, pos=0.65, inner sep=1pt, text=agB]{\small$a_2$} (o);
  \draw[arr, draw=agA] (ldr) -- node[above right, pos=0.65, inner sep=1pt, text=agA]{\small$a_1$} (o);
  \draw[arr] (o) -- node[above]{\small$s'$} (p1b);
  \draw[arr] (o) -- node[above]{\small$s'$} (p2b);
\end{tikzpicture}}\hfill
\subfigure[Opinion-guided implementation]{%
\hspace*{4mm}\begin{tikzpicture}[node distance=1.4cm]
  \node[gstate] (s) {$\mathcal{G}$};
  \node[snodeA, above right=0.5cm and 0.3cm of s] (ldr) {$\Opol_1(s, \tilde{z}_1)$};
  \node[snodeB, below right=0.5cm and 0.3cm of s] (flw) {$\Opol_2(s, \tilde{z}_2)$};
  \node[gstate, right=1.8cm of s]                  (o)  {$\mathcal{G}$};
  \node[rectangle, above right=0.5cm and 0.3cm of o] (p1b) {$\ldots$};
  \node[rectangle, below right=0.5cm and 0.3cm of o] (p2b) {$\ldots$};

  \draw[arr] (s) -- node[above left]{\small$s$} (ldr);
  \draw[arr] (s) -- node[below left]{\small$s$} (flw);
  \draw[arr, draw=agA, transform canvas={xshift=-3pt}] (ldr) -- node[left, xshift=2pt, font=\scriptsize, text=agA]{$\tilde{z}_1$} (flw);
  \draw[arr, draw=agB, transform canvas={xshift=3pt}]  (flw) -- node[right, xshift=-2pt, font=\scriptsize, text=agB]{$\tilde{z}_2$} (ldr);
  \draw[arr, draw=agA] (ldr) -- node[above right, pos=0.65, inner sep=1pt, text=agA]{\small$a_1$} (o);
  \draw[arr, draw=agB] (flw) -- node[below right, pos=0.65, inner sep=1pt, text=agB]{\small$a_2$} (o);
  \draw[arr] (o) -- node[above]{\small$s'$} (p1b);
  \draw[arr] (o) -- node[above]{\small$s'$} (p2b);
\end{tikzpicture}\hspace*{4mm}}
	\caption{Four decision structures for a two-agent game. (a) agents act concurrently with independent strategies $\sigma_1,\sigma_2$. (b) a reactive strategy realized sequentially, where agent~2 observes $a_1$ before selecting $a_2$. (c) mutually reactive strategies, where each agent's action conditions on the other's; this is an abstract specification without a realization mechanism. (d) an opinion-guided implementation, where each agent maintains an opinion state $\tilde{z}_i$ paired to the other's through $\Gamma_i$, and selects its action via $\Opol_i(s, \tilde{z}_i)$.
	\label{fig:decision_structure}}
\end{figure*}

So far we have assumed that the two agents select their actions concurrently and independently, as depicted in panel~(a) of Fig.~\ref{fig:decision_structure}. A common alternative is a sequential (or turn-based) decision structure, shown in panel~(b), in which the agents act one after the other within each step. The second agent observes the first's action and conditions its decision on the history $h$ together with the observed action. The resulting strategy is called a \emph{reactive strategy} and corresponds to the notion of a reaction function in the game-theoretic literature~\cite{bacsar1998dynamic}\footnote{The \emph{reactive} strategy in this paper conditions on the other agent's \emph{action}, unlike reactive strategies in control and robotics, which condition on the \emph{state} of the environment~\cite{kappler2018real,kress2009temporal}.}.

\begin{definition}[Reactive Strategy]\label{def:reactive}
Given a game $\mathcal{G}$, a reactive strategy for agent $i\!\in\!\{1,2\}$ is a mapping $\Apol_i : \mathcal{H} \times A_{j} \rightarrow A_i$, where $\Apol_i(h, a_{j}) \!=\! a_i$. The set of all such strategies is denoted $\Sigma^{\text{act}}_i$. 
\end{definition}
\noindent In contrast, a strategy $\sigma_i: \mathcal{H}\!\rightarrow\!A_i$ that does not condition on the other agent's action is called a \emph{non-reactive strategy}. If only one agent applies a reactive strategy while the other uses a non-reactive strategy, as in panel~(b) of Fig.~\ref{fig:decision_structure}, the joint strategy is defined as
\begin{equation}
(\sigma_1\,\|\,\Apol_{2})(h)\!:=\!\{(\sigma_1(h),\Apol_{2}(h,\sigma_1(h)))\}.
\end{equation}
Consider now the mutually reactive setting, in which both agents employ reactive strategies and each conditions on the other's action, as depicted in panel~(c) of Fig.~\ref{fig:decision_structure}. In this setting, the joint strategy is no longer single-valued but set-valued, defined as
$(\Apol_1\,\|\,\Apol_{2}): \mathcal{H} \rightarrow 2^{A_1 \times A_2}$,
whose value at a history $h$ collects all action pairs in which each agent's action is a reaction to the other's,
\begin{equation}\label{eq:reactive_joint}
\begin{aligned}
(\Apol_1\,\|\,\Apol_{2})(h) := \bigl\{(a_1, a_2) \,|\, a_1 &= \Apol_1(h,\! a_2), \\ a_2 &= \Apol_2(h,\! a_1)\bigr\}.
\end{aligned}
\end{equation}
Depending on $\Apol_1$ and $\Apol_2$, this set may be empty, a singleton, or contain several compatible action pairs. 

Non-reactive and reactive strategies are both \emph{individual strategies}: each returns agent~$i$'s own action $a_i\!\in\!A_i$, in contrast to a joint strategy that returns an action pair for both agents. For any two individual strategies, the joint outcome $(\sigma_i\,\|\,\sigma_{j})(h)\!\subseteq\!A_1\!\times\!A_2$ collects all joint action pairs realized by the two agents. The $\pi$-homogeneous implementability of Def.~\ref{def:robust_coord_homo} is understood in this general sense.

\begin{tcolorbox}[runningexample]
\textbf{Running example}: Recall that at the symmetric initial state $s_0\!=\!(p_2, p_2)$, the permissive joint strategy is $\pi(s_0)\!=\!\{(a_{p_1}, a_{p_3}), (a_{p_3}, a_{p_1})\}$. A reactive strategy can be designed as
\begin{equation}~\label{eq:pi_reactive}
\Apol(s_0, a_{j}) = \begin{cases}
a_{p_1} & \text{if } a_{j} = a_{p_3}, \\
a_{p_3} & \text{if } a_{j} = a_{p_1},
\end{cases}
\end{equation}
where $a_{j}$ denotes the action of the other agent. If one agent applies the reactive strategy while the other applies an arbitrary rational strategy $\sigma_{j}$, then $\sigma_{j}(s_0)\!\in\!A_j^\rat(s_0)\!=\!\{a_{p_1},a_{p_3}\}$, and the reactive response yields the pair $(a_{p_3}, a_{p_1})$ when $\sigma_{j}(s_0)\!=\!a_{p_1}$ and $(a_{p_1}, a_{p_3})$ when $\sigma_{j}(s_0)\!=\!a_{p_3}$, both of which lie in $\pi(s_0)$. Thus, the reactive strategy can achieve robust coordination. When both agents apply this reactive strategy, the resulting set of action pairs is $(\Apol\,\|\,\Apol\circ\bar\tau)(s_0)\!=\!\{(a_{p_3},a_{p_1}), (a_{p_1},a_{p_3})\}\!=\!\pi(s_0)$: each agent reacts to the other, so the only mutually consistent outcomes are those in which one plays $a_{p_1}$ and the other plays $a_{p_3}$. This bypasses the issue caused by symmetry noted earlier, where no homogeneous individual strategy could produce any pair in $\pi(s_0)$.
\end{tcolorbox}

The above example illustrates that reactive strategies can deliver both robust coordination and homogeneous implementation. However, reactive strategies are built on the assumption that each agent can observe the other's intended action before execution, which is unrealistic. The remaining question is how the two agents converge on a common pair in $(\Apol_1\,\|\,\Apol_2)(h)$ without knowing each other's intended actions in advance. Inspired by the properties of opinion dynamics introduced in Section~\ref{sec:nod}, we propose to realize reactive strategies using opinion dynamics.

\subsection{Opinion-Guided Reactive Strategies}

We illustrate this opinion-guided implementation, depicted in panel~(d) of Fig.~\ref{fig:decision_structure}, on the running example.

\begin{tcolorbox}[runningexample]
\textbf{Running example}: at the initial state $s_0$, we equip each agent with a scalar opinion $z_i\!\in\!\mathbb{R}$ evolving under the following nonlinear opinion dynamics:
\begin{equation}\label{eq:od_running}
\dot{z}_i = -3 z_i + 3\tanh\!\left(3(z_i - z_j)\right),
\end{equation}
where $i,j\!\in\!\{1,2\}$ and $i\!\neq\!j$. The individual strategies of both agents follow the same rule
\begin{equation}\label{eq:op_running_example}
\sigma^{\mathrm{op}}(s_0, z_i) = \begin{cases}
a_{p_1} & \text{if } z_i \geq 0, \\
a_{p_3} & \text{if } z_i < 0,
\end{cases}
\end{equation}
 which is homogeneous by construction. The discrete decision is made after a reserved interval $\Delta t_{\mathrm{op}}\!>\!0$, during which the opinion dynamics~\eqref{eq:od_running} are expected to converge. By~\cite[Cor.~IV.1.2]{bizyaeva2022nonlinear}, the neutral equilibrium $(z_1,z_2)\!=\!(0,0)$ is unstable, and two stable equilibria emerge at $(z^{*},-z^{*})$ and $(-z^{*},z^{*})$, with $z^{*}\!\approx\!1$, as shown in the phase portrait of Fig.~\ref{fig:od_running_evolution}(a). Provided $\Delta t_{\mathrm{op}}$ is sufficiently large, the trajectories converge to one of these equilibria within the interval (Fig.~\ref{fig:od_running_evolution}(b)). Under the rule~\eqref{eq:op_running_example}, the joint action ultimately settles on either $(a_{p_1}, a_{p_3})$ or $(a_{p_3}, a_{p_1})$, both of which lie in $\pi(s_0)$ by~\eqref{eq:pi_reactive}.
\noindent\includegraphics[width=1.0\linewidth]{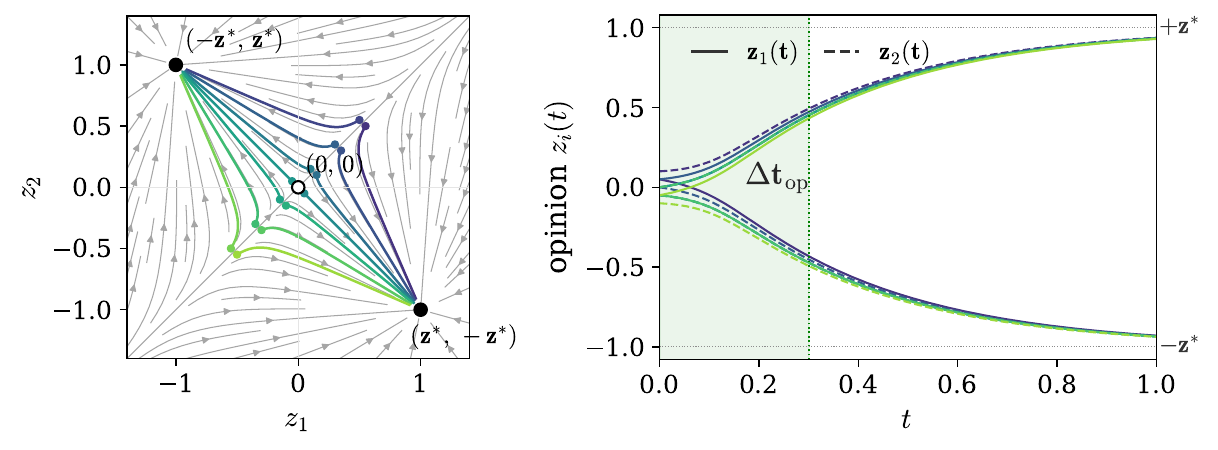}
\captionof{figure}{(a) Phase portrait of~\eqref{eq:od_running}: trajectories from generic initial conditions diverge from $(0,0)$ (open circle) and converge to one of the two stable equilibria (filled circles). (b) Time evolution of $z_1(t)$ (solid) and $z_2(t)$ (dashed) from several initial conditions; the shaded interval marks the reserved window $\Delta t_{\mathrm{op}}$.}
\label{fig:od_running_evolution}
\end{tcolorbox}



The example shows that the decentralized coordination challenges are mitigated by the opinion dynamics. The instability of the symmetric equilibrium $(0,0)$ breaks the initial symmetry between the two homogeneous agents and drives their opinions apart, assigning opposite roles in this anti-coordination scenario. Both stable equilibria pair the two opinions with opposite signs, so each agent is committed to the role complementary to the other's, and its choice thereby becomes adaptive to the other agent's unknown decision behavior.

This opinion-guided implementation is intrinsically two-layered, with the layers operating on separated time scales. On the \emph{continuous opinion layer}, each agent's opinion $z_i$ evolves under mutual coupling and converges to one of the stable equilibria. On the \emph{discrete strategy layer}, the sign of $z_i$ commits the agent to a discrete action through the rule $\sigma^{\mathrm{op}}$ in~\eqref{eq:op_running_example}. In what follows, we first formalize the opinion-guided strategy at the discrete strategy layer, generalizing it from the scalar running example to multi-dimensional opinions, and then present its layered realization through the continuous opinion dynamics, sketched in Fig.~\ref{fig:layered}.

\begin{figure}
	\centering
	\input{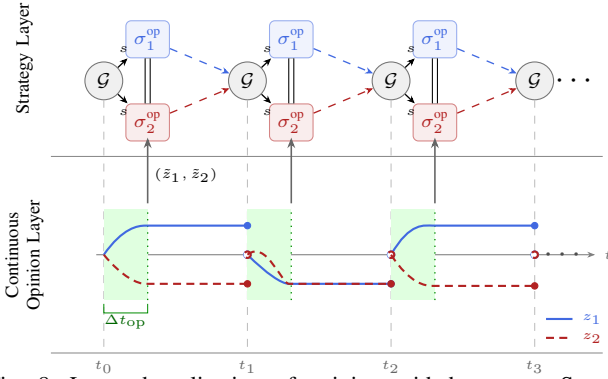}
	\vskip-.3cm
	\caption{Layered realization of opinion-guided strategy. Strategy layer: at each decision time $t_k$, both agents query the game $\mathcal{G}$ at the current state $s$ and select opinion-guided strategies $\Opol_1, \Opol_2$ whose action intents are compatible, as required by~\eqref{eq:opinion_joint}. Continuous opinion layer: between consecutive decisions, the coupled opinions $z_1$ and $z_2$ converge to opposite signs within the reserved window $\Delta t_{\mathrm{op}}$. At the end of each window, the converged opinion is abstracted to its discrete sign $\tilde z_i$ and passed up to the strategy layer, which commits one of the compatible action pairs.}
	\label{fig:layered}
 \end{figure}

\subsubsection{Opinion-Guided Strategy}
As the running example suggests, one-dimensional opinion dynamics with a bistable bifurcation reliably converge to one of two stable equilibria~\cite{bizyaeva2022nonlinear}. When the opinion state is multi-dimensional, each component behaves analogously, so we abstract each agent $i$'s continuous opinion into a discrete opinion state $\tilde{z}_i\!\in\!\mathbb{B}^{n_i}$, where $\mathbb{B}\!:=\!\{-1,1\}$ is a binary index set and $n_i\!\in\!\mathbb{N}$ is the dimension of that agent's opinion state. Each entry of $\tilde{z}_i$ records which of the two subsets has been selected along the corresponding component, and this discrete abstraction induces a partition of the agent's action space. For each agent $i\!\in\!\{1,2\}$, we define a map $\O_{i}\!:\!A_{i}\!\to\!\mathbb{B}^{n_i}$ from actions to discrete opinion states. It labels each action with an element of $\mathbb{B}^{n_i}$, thereby partitioning $A_{i}$ into at most $2^{n_i}$ subsets; we write $A_i^{(\tilde{z}_i)}\!:=\!\{a\!\in\!A_{i}\mid\O_i(a)\!=\!\tilde{z}_i\}$ for the set of actions labeled $\tilde{z}_i$. Given the partitioning, the opinion state of agent $i$ is used to select an action that responds appropriately to the action agent $j$ intends. This motivates the following definition of opinion-guided strategies.

\begin{definition}[Opinion-Guided Strategy]\label{def:opinion_guided}
Given a game $\mathcal{G}$ and partition maps $\O_1,\O_2$, an opinion-guided strategy for agent $i\!\in\!\{1,2\}$ is a mapping $\Opol_i\!:\!\mathcal{H}\!\times\!\mathbb{B}^{n_i}\!\rightarrow\!A_i$ from a history $h$ and a discrete opinion state $\tilde{z}_i$ to an action $a_i$. Agent $i$'s opinion state is determined from its opponent's through its own \emph{opinion-pairing map} $\Gamma_i\!:\!\mathcal{H} \times \mathbb{B}^{n_j} \!\rightarrow\! \mathbb{B}^{n_i}$:
\begin{equation}\label{eq:opinion_coupling}
a_i = \Opol_i(h,\tilde{z}_{i}), \quad
\tilde{z}_i = \Gamma_i(h, \tilde{z}_{j}),
\end{equation}
where the selected action lies in $A_i^{(\tilde{z}_i)}$ with $\tilde{z}_i\!=\!\O_i(a_i)$, and $\tilde{z}_{j}\!=\!\O_{j}(a_{j})$ is the opinion state of the other agent $j\!\in\!\{1,2\}\!\setminus\!\{i\}$ corresponding to its intended action $a_j$. The set of all such strategies is denoted $\Sigma^{\text{op}}_i$.
\end{definition}

Each opinion-pairing map $\Gamma_i$ encodes how agent $i$ reacts to the other's opinion state. For simplicity, we assume in this paper that $\Gamma_i$ is linear in the opinion state $\tilde{z}_j$ and that each component of $\tilde{z}_i$ is coupled to at most one component of $\tilde{z}_j$.

\begin{assumption}[Linear Opinion Pairing]\label{asm:linear_pairing}
The opinion-pairing map at $h$ is linear, represented by a \emph{partial signed permutation} matrix $\Gamma_i(h)\!\in\!\{-1,0,+1\}^{n_i\times n_j}$ with at most one nonzero entry per row.
\end{assumption}
\noindent Specifically, the $k$-th row of $\Gamma_i(h)$ governs the $k$-th component of $\tilde{z}_i$: a $+1$ entry couples it in agreement with the matched component of $\tilde{z}_j$, a $-1$ entry couples it in opposition, and an all-zero row leaves it free, coupled to no component of $\tilde{z}_j$. In the running example of Section~\ref{sec:example}, each agent's opinion-pairing map at $s_0$ reduces to the scalar $\Gamma_i(s_0)\!=\!-1$, representing the rule ``if the other agent commits to one target zone, I commit to the other''. To achieve the coordination required by $\pi$, the maps $\Gamma_1$ and $\Gamma_2$ must be designed carefully, which will be discussed in Section~\ref{sec:analysis}.

When agent $1$ employs an opinion-guided strategy $\Opol_1\!\in\!\Sigma^{\text{op}}_1$ and agent $2$ a non-reactive strategy $\sigma_2\!:\!\mathcal{H}\!\rightarrow\!A_2$, the joint outcome is uniquely determined as $(\Opol_1\,\|\,\sigma_2)(h) = \{(\Opol_1(h,\tilde{z}_1),\, a_2)\}$, where $a_2\!=\!\sigma_2(h)$, $\tilde{z}_2\!=\!\O_2(a_2)$, and $\tilde{z}_1\!=\!\Gamma_1(h)\,\tilde{z}_2$. When both agents employ opinion-guided strategies $\Opol_1,\Opol_{2}$, their opinion states must jointly satisfy both couplings $\tilde{z}_1\!=\!\Gamma_1(h)\,\tilde{z}_2$ and $\tilde{z}_2\!=\!\Gamma_2(h)\,\tilde{z}_1$, which in general admit multiple solutions, so the joint outcome is set-valued:
\begin{equation}\label{eq:opinion_joint}
\begin{aligned}
(\Opol_1\,\|\,\Opol_{2})(h) \!:=\! \big\{ & (a_1,a_2) \mid
\exists \tilde{z}_i,\tilde{z}_j:\ a_i\!=\!\Opol_i(h,\tilde{z}_i), \\
& \tilde{z}_i\!=\!\Gamma_i(h)\,\tilde{z}_j,\ \tilde{z}_j\!=\!\O_j(a_j), \\
& i,j\!\in\!\{1,2\},\ i\!\neq\!j \,\big\}.
\end{aligned}
\end{equation}
To realize such joint behavior, the agents need to exchange their opinion states and converge on common pairs of opinions that satisfy the coupling constraints. This is achieved by the continuous opinion dynamics described subsequently.

\begin{remark}\label{rem:fixed_partition}
The fixed partition $\O_{i}\!:\!A_{i}\!\to\!\mathbb{B}^{n_i}$ is restrictive: since $\pi(h)$ varies across histories, a history-dependent partition would be more flexible. For simplicity, we adopt the fixed partition in this paper.
\end{remark}

\subsubsection{Layered Realization} To realize the opinion-guided strategy in Def.~\ref{def:opinion_guided}, agent $i$ maintains a continuous opinion state $z_i(t)\!\in\!\mathbb{R}^{n_i}$ that evolves according to the NOD equation~\cite{leonard2024fast} and is mapped to its discrete counterpart $\tilde{z}_i\!\in\!\mathbb{B}^{n_i}$ via the element-wise sign function $\tilde{z}_i\!:=\!\mathrm{sgn}(z_i)$. Before introducing the communication-free implementation (Section~\ref{sec:case_study}), we first describe the idealized interaction in which each agent receives its opponent's opinion state directly in real time through communication. The layered realization of the opinion-guided strategy~\eqref{eq:opinion_coupling} is given by:

\begin{equation}\label{eq:layered_opinion}
\hspace{-1em}\left\{
\begin{aligned}
\text{Strategy layer:}\, a_i & \!=\! \Opol_i(h, \tilde{z}_i), \; \tilde{z}_i \!=\! \mathrm{sgn}(z_{i}(t_k \!+\! \Delta t_{\mathrm{op}})) \\
\text{Opinion layer:}\, \dot{z}_i & \!=\! -d z_i \!+\! \rho_i   \tanh\!\left(\nu z_i \!+\! \beta  \Gamma_i(h) z_{j} \!+\! b_i\right)
\end{aligned}
\right.
\end{equation}
where $d\!>\!0$ is a damping coefficient, $\tanh$ is a saturation function that enables fast and flexible decision-making, and $\rho_i\!\geq\!0$ is a tunable attention parameter. The gains $\nu\!>\!0$ and $\beta\!>\!0$ weigh the self- and opponent-coupling terms, and $b_i$ is a constant bias encoding agent~$i$'s prior preference. Here $\Gamma_i(h)\,z_{j}\!\in\!\mathbb{R}^{n_i}$ is the target opinion that agent $i$ should align with, where $\Gamma_i(h)$ is the opinion-pairing map from Def.~\ref{def:opinion_guided}. The two layers evolve on separated time scales: the strategy layer acts in discrete time, at decision instants $t_k$ indexed by $k\!\in\!\mathbb{N}$, while the opinion layer evolves in continuous time $t\!\in\!\mathbb{R}_{\geq0}$. Each decision step begins at $t_k$ with a reserved window of length $\Delta t_{\mathrm{op}}$ over which the opinion converges; the action is then committed from the settled opinion at $t_k\!+\!\Delta t_{\mathrm{op}}$ and executed for the remainder of the step.

\section{Analysis for Opinion-Guided Coordination}\label{sec:analysis}

In this section, we analyze formal guarantees for the layered opinion-guided strategy. The analysis is organized layer by layer. We first analyze the strategy layer, characterizing how to design an opinion-guided strategy that achieves robust coordination against rational opponents and admits homogeneous deployment. We then analyze the continuous opinion layer and close with a layered guarantee, showing that the nonlinear opinion dynamics realize the desired coordination. 

\subsection{Robust Coordination at Strategy Layer}

For each $i\!\in\!\{1,2\}$, the function $\O_i\!: \!A_i\!\to\!\mathbb{B}^{n_i}$ partitions the action set $A_i$ into the subsets $A_i^{(\tilde{z}_i)}$ defined earlier. Next, we introduce the partition-restricted robust action sets $A_i^\rob(h;\tilde{z}_{j}, \Sigma_j)$, which refine $A_i^\rob(h;\Sigma_j)$ in Def.~\ref{def:robust_set} by quantifying over a single cell $A_{j}^{(\tilde{z}_j)}$ of agent $j$ rather than the full selectable set $A^{\Sigma_j}(h)$, as follows:
\begin{equation}\label{eq:rob_partition}
\begin{aligned}
A_i^\rob(h;\tilde{z}_{j}, \Sigma_j) \!:=\! \{ a_i \;\big|\; & \forall a_{j}\!\in\!A_{j}^{(\tilde{z}_{j})}\!\cap\! A^{\Sigma_j}(h): \\
& (a_1,a_2)\!\in\!\pi(h) \}.
\end{aligned}
\end{equation}
Two degenerate cases make agent $i$'s choice immaterial. If $A_{j}^{(\tilde{z}_{j})} \cap A^{\Sigma_j}(h)\!=\!\emptyset$, agent $j$ has no selectable action in the cell $A_{j}^{(\tilde{z}_{j})}$, so there is nothing to coordinate against. If $\pi(h)\!=\!\emptyset$, no joint action is required at $h$. In both cases, we set $A_i^\rob(h;\tilde{z}_{j}, \Sigma_j)\!=\!A_i$, so the robust set imposes no restriction. When the selectable set is fixed to the rational set $\Sigma_j^\rat$, we abbreviate $A_i^\rob(h;\tilde{z}_{j})\!:=\!A_i^\rob(h;\tilde{z}_{j},\Sigma_j^\rat)$.

To illustrate how the partitions induced by $\O_1,\O_2$ enable robust coordination, we consider the permissive action set $\pi_4(h)$ in Fig.~\ref{fig:coordination_analysis}(d), in which neither agent has a non-empty robust action set, i.e., $A_i^\rob(h)\!=\!\emptyset$ for both $i\!\in\!\{1,2\}$. As shown in Fig.~\ref{fig:partition}(a), partitioning $A_{1}$ via $\O_{1}$ makes $A_2^\rob(h;\tilde{z}_{1})\!\neq\!\emptyset$ on each cell $A_{1}^{(\tilde{z}_{1})}$, $\tilde{z}_{1}\!\in\!\{-1,1\}$, so agent $2$ can guarantee coordination on its own once $\tilde{z}_{1}$ is known. As shown in Fig.~\ref{fig:partition}(b), when both agents partition their action sets via $\O_1$ and $\O_2$, $A_i^\rob(h;\tilde{z}_{j})\!\neq\!\emptyset$ holds for each agent $i$ and every opinion state $\tilde{z}_j$ of the other agent, so each agent can guarantee coordination on its own once the other agent's opinion state is known. The conditions under which an opinion-guided strategy achieves $\pi$-robust coordination are formalized in the following theorems.

\begin{figure}[htb]
    \centering
    \includegraphics[width=\linewidth]{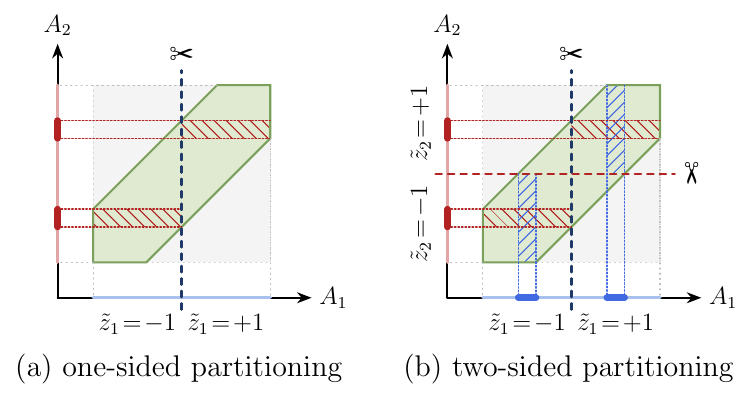}
    \caption{Illustration of the opinion-guided strategy from the perspective of agent $i$, $i\!\in\!\{1,2\}$. The dashed line with a scissor icon denotes the partition of $A_{j}$ induced by $\O_{j}$. In~(a), only $A_{j}$ is partitioned, indexed by $\tilde{z}_{j}\!\in\!\{-1,1\}$. In~(b), the action sets of both agents are mutually partitioned according to their respective discrete opinion states $\tilde{z}_1,\tilde{z}_2\!\in\!\mathbb{B}^n$.}
    \label{fig:partition}
\end{figure}



\begin{theorem}[Robust coordination against a set of non-reactive opponents]\label{th:op_coord_1}
Consider a two-agent game $\mathcal{G}$ with permissive joint strategy $\pi$, given partition functions $\O_{i}, \O_{j}$, and a set of non-reactive rational strategies $\Sigma_j \!\subseteq\! \Sigma_j^{\rat}$ for agent $j$. There exists an opinion-guided strategy $\Opol_i\!\in\!\Sigma_i^{\text{op}}$ for agent $i$ that achieves $(\pi, \Sigma_{j})$-robust coordination \textbf{if} there exists an opinion-pairing map $\Gamma_i$ such that for every $h\!\in\!\mathcal{H}$ and every opinion state $\tilde{z}_{j}\!\in\!\mathbb{B}^{n_j}$:
\begin{equation}\label{eq:coord1_cond}
A_i^\rob(h;\tilde{z}_{j}, \Sigma_j) \cap A_i^{(\tilde{z}_i)}\!\neq\!\emptyset, \quad \tilde{z}_i\!=\!\Gamma_i(h,\tilde{z}_{j}).
\end{equation}
\end{theorem}

\begin{proof}
Let $\Gamma_i$ be an opinion-pairing map satisfying~\eqref{eq:coord1_cond}. We first build an opinion-guided strategy with $\Gamma_i$ and then show it is robust. For every history $h\!\in\!\mathcal{H}$ and opponent opinion $\tilde{z}_{j}\!\in\!\mathbb{B}^{n_j}$, let $\tilde{z}_i\!=\!\Gamma_i(h,\tilde{z}_{j})$ and choose $\Opol_i(h,\tilde{z}_i)\!\in\!A_i^\rob(h;\tilde{z}_{j}, \Sigma_j)\cap A_i^{(\tilde{z}_i)}$, which is non-empty by~\eqref{eq:coord1_cond}. Next, we show that the designed $\Opol_i$ is $(\pi,\Sigma_{j})$-robust. Let $\sigma_{j}\!\in\!\Sigma_{j}$ and $h\!\in\!\mathcal{H}$ with $\pi(h)\!\neq\!\emptyset$ be arbitrary. The opponent plays $a_{j}\!=\!\sigma_{j}(h)$, which reveals its opinion $\tilde{z}_{j}\!=\!\O_{j}(a_{j})$ and places $a_{j}\!\in\!A_{j}^{(\tilde{z}_{j})}\!\cap\!A^{\Sigma_j}(h)$. Agent $i$ reacts with $\tilde{z}_i\!=\!\Gamma_i(h,\tilde{z}_{j})$ and plays $a_i\!=\!\Opol_i(h,\tilde{z}_i)$, which by construction lies in $A_i^\rob(h;\tilde{z}_{j}, \Sigma_j)$. By the definition of the robust action set~\eqref{eq:rob_partition}, such an action coordinates with every action agent $j$ may select in the cell $A_{j}^{(\tilde{z}_{j})}$, in particular with $a_{j}$, so $(a_i,a_{j})\!\in\!\pi(h)$. As $\sigma_{j}$ and $h$ were arbitrary, coordination holds against every rational strategy in $\Sigma_{j}$ at every history with $\pi(h)\!\neq\!\emptyset$, i.e.\ $\Opol_i$ is $(\pi,\Sigma_{j})$-robust.
\end{proof}


Theorem~\ref{th:op_coord_1} gives a sufficient condition for a $(\pi, \Sigma_{j})$-robust opinion-guided strategy to exist: there is an opinion-pairing map $\Gamma_i$ under which, whenever agent $j$ commits to an intent $\tilde{z}_{j}$, the paired opinion state $\tilde{z}_i\!=\!\Gamma_i(h,\tilde{z}_{j})$ of agent $i$ contains an action that handles every action agent $j$ may select in $A_{j}^{(\tilde{z}_{j})} \!\cap\! A^{\Sigma_j}(h)$. As shown in the proof, the strategy is then constructed by letting agent $i$ commit to an action of $A_i^\rob(h;\tilde{z}_{j}, \Sigma_j)\cap A_i^{(\tilde{z}_i)}$.

Theorem~\ref{th:op_coord_1} holds for any fixed $\Sigma_j\!\subseteq\!\Sigma_j^\rat$; taking $\Sigma_{j}\!=\!\Sigma_{j}^{\rat}$, the full rational strategy set of agent $j$, specializes it to the strongest guarantee, robustness against every non-reactive rational opponent at once.

\begin{theorem}[Robust coordination against all rational opponents]\label{th:op_coord_full}
Consider a two-agent game $\mathcal{G}$ with permissive joint strategy $\pi$ and given partition functions $\O_{i},\O_{j}$. There exists an opinion-guided strategy $\Opol_i\!\in\!\Sigma_i^{\text{op}}$ for agent $i$ that achieves $(\pi,\Sigma_{j}^{\rat})$-robust coordination \textbf{if and only if} there exists an opinion-pairing map $\Gamma_i$ such that for every $h\!\in\!\mathcal{H}$ and every opinion state $\tilde{z}_{j}\!\in\!\mathbb{B}^{n_j}$:
\begin{equation}\label{eq:coordfull_cond}
A_i^\rob(h;\tilde{z}_{j}) \cap A_i^{(\tilde{z}_i)}\!\neq\!\emptyset, \quad \tilde{z}_i\!=\!\Gamma_i(h,\tilde{z}_{j}).
\end{equation}
\end{theorem}
\begin{proof}
\emph{($\Leftarrow$) Sufficiency.} This follows from Theorem~\ref{th:op_coord_1} with $\Sigma_{j}\!=\!\Sigma_{j}^{\rat}$.

\emph{($\Rightarrow$) Necessity.} We argue by contradiction. Suppose $\Opol_i\!\in\!\Sigma_i^{\text{op}}$ is $(\pi,\Sigma_{j}^{\rat})$-robust with pairing $\Gamma_i$, yet at some history $h^\star\!\in\!\mathcal{H}$ and opinion state $\tilde{z}_{j}^\star$, $A_i^\rob(h^\star;\tilde{z}_{j}^\star)\cap A_i^{(\tilde{z}_i^\star)}\!=\!\emptyset$, where $\tilde{z}_i^\star\!=\!\Gamma_i(h^\star,\tilde{z}_{j}^\star)$. Consider the action $a_i^\star\!=\!\Opol_i(h^\star,\tilde{z}_i^\star)$, which lies in $A_i^{(\tilde{z}_i^\star)}$ by Def.~\ref{def:opinion_guided}. For every rational action $a_{j}\!\in\!A_{j}^{(\tilde{z}_{j}^\star)}\!\cap\!A_{j}^\rat(h^\star)$, some $\sigma_{j}\!\in\!\Sigma_{j}^{\rat}$ plays $a_{j}$ at $h^\star$; since $\O_{j}(a_{j})\!=\!\tilde{z}_{j}^\star$, agent $i$ reacts with $\tilde{z}_i^\star$ and plays $a_i^\star$, and robustness gives $(a_i^\star,a_{j})\!\in\!\pi(h^\star)$. By~\eqref{eq:rob_partition}, $a_i^\star\!\in\!A_i^\rob(h^\star;\tilde{z}_{j}^\star)$, so $a_i^\star\!\in\!A_i^\rob(h^\star;\tilde{z}_{j}^\star)\cap A_i^{(\tilde{z}_i^\star)}$, contradicting the emptiness of this intersection. Hence, if there exists
a $(\pi,\Sigma_{j}^{\rat})$-robust opinion-guided strategy, then no pair $(h^\star,\tilde{z}_{j}^\star)$ with $A_i^\rob(h^\star;\tilde{z}_{j}^\star)\cap A_i^{(\tilde{z}_i^\star)}\!=\!\emptyset$ exists. Thus, the necessity is proved.
\end{proof}

\begin{tcolorbox}[runningexample]
\textbf{Running example}: Recall that $\pi(s_0)\!=\!\{(a_{p_1}, a_{p_3}), (a_{p_3}, a_{p_1})\}$ and $A_i^\rat(s_0)\!=\!\{a_{p_1}, a_{p_3}\}$. Let each robot hold a scalar opinion state ($n_1\!=\!n_2\!=\!1$), and let both share the partition function $\O(a_{p_1})\!=\!+1$, $\O(a_{p_2})\!=\!\O(a_{p_3})\!=\!-1$ and the opinion-pairing map $\Gamma(s_0)\!=\!-1$. For $\tilde{z}_{j}\!=\!+1$, the cell $A_{j}^{(+1)}\!=\!\{a_{p_1}\}$ contains the single rational action $a_{p_1}$, and only $a_{p_3}$ coordinates with it, so $A_i^\rob(s_0;+1)\!=\!\{a_{p_3}\}$ by~\eqref{eq:rob_partition}; the paired opinion state $\tilde{z}_i\!=\!-1$ then gives $A_i^\rob(s_0;+1)\cap A_i^{(-1)}\!=\!\{a_{p_3}\}\!\neq\!\emptyset$. The same computation for $\tilde{z}_{j}\!=\!-1$ gives $A_i^\rob(s_0;-1)\cap A_i^{(+1)}\!=\!\{a_{p_1}\}\!\neq\!\emptyset$. Condition~\eqref{eq:coord1_cond} thus holds at $s_0$, and the robust opinion-guided strategy can be designed as
\begin{equation}\label{eq:example_strategy}
	\Opol(s_0, +1)\!=\!a_{p_1}, \quad \Opol(s_0, -1)\!=\!a_{p_3}.
\end{equation}
With this strategy, each robot can guarantee coordination at $s_0$ with any rational opponent, regardless of which rational action it selects, as established by Theorem~\ref{th:op_coord_full}. 
\end{tcolorbox}

We now extend Theorem~\ref{th:op_coord_1} to the setting in which both agents employ opinion-guided strategies. Extending Def.~\ref{def:rational}, we call an opinion-guided strategy $\Opol_j\!\in\!\Sigma_j^{\text{op}}$ \emph{rational} if $\Opol_j(h,\tilde{z}_{j})\!\in\!A_j^\rat(h)$ for every $h\!\in\!\mathcal{H}$ and every $\tilde{z}_{j}\!\in\!\mathbb{B}^{n_j}$, and we denote by $\Sigma_{j}^{\text{op},\rat}$ the set of all such strategies. The following theorem characterizes robust coordination against $\Sigma_{j}^{\text{op},\rat}$.

\begin{theorem}[Robust coordination against opinion-guided opponents]~\label{th:op_coord_2}
For a two-agent game $\mathcal{G}$ with permissive joint strategy $\pi$, suppose the partition functions $\O_i,\O_{j}$ and the other agent's opinion-pairing map $\Gamma_j$ are given. There exists an opinion-guided strategy $\Opol_i\!\in\!\Sigma_i^{\text{op}}$ for agent $i$ that achieves $(\pi,\Sigma_j^{\text{op},\rat})$-robust coordination if there exists an opinion-pairing map $\Gamma_i$ such that, for every $h\!\in\!\mathcal{H}$ with $\pi(h) \!\neq\! \emptyset$, the coupling
\begin{equation}~\label{eq:coupling}
\begin{bmatrix}
	I_{n_i} & -\Gamma_i(h) \\
	-\Gamma_j(h) & I_{n_j}
\end{bmatrix} \begin{bmatrix}
\tilde{z}_i \\
\tilde{z}_{j}
\end{bmatrix} = 0
\end{equation}
admits at least one solution $(\tilde{z}_i,\tilde{z}_{j})$ with $A_{j}^{(\tilde{z}_{j})}\!\cap\!A_{j}^\rat(h)\!\neq\!\emptyset$, and every such solution satisfies:
\begin{equation}\label{eq:coord2_cond}
A_i^\rob(h;\tilde{z}_{j}) \cap A_i^{(\tilde{z}_i)}\!\neq\!\emptyset.
\end{equation}
Here, $I_{n_i}$ and $I_{n_j}$ denote the identity matrices of size $n_i$ and $n_j$.\end{theorem}

\begin{proof}
Let $\Gamma_i$ be an opinion-pairing map satisfying the conditions of the theorem. We first build an opinion-guided strategy $\Opol_i$ from $\Gamma_i$ and then show it coordinates robustly with every rational opinion-guided opponent. For every history $h\!\in\!\mathcal{H}$ and opinion state $\tilde{z}_i\!\in\!\mathbb{B}^{n_i}$, let $\tilde{z}_{j}\!=\!\Gamma_j(h)\tilde{z}_i$ be the opinion state paired with $\tilde{z}_i$ through the opponent's map. If $(\tilde{z}_i,\tilde{z}_{j})$ solves the coupling~\eqref{eq:coupling} with $A_{j}^{(\tilde{z}_{j})}\!\cap\!A_{j}^\rat(h)\!\neq\!\emptyset$, then $\tilde{z}_i\!=\!\Gamma_i(h)\tilde{z}_{j}$ and we choose
\begin{equation}\label{eq:coord2_construct}
\Opol_i(h,\tilde{z}_i)\!\in\!A_i^\rob(h;\tilde{z}_{j}) \cap A_i^{(\tilde{z}_i)},
\end{equation}
which is non-empty by~\eqref{eq:coord2_cond}; otherwise let $\Opol_i(h,\tilde{z}_i)$ be any action of $A_i^{(\tilde{z}_i)}$.

Next, we show that the designed $\Opol_i$ is $(\pi,\Sigma_{j}^{\text{op},\rat})$-robust. Let $\Opol_{j}\!\in\!\Sigma_{j}^{\text{op},\rat}$ and $h\!\in\!\mathcal{H}$ with $\pi(h)\!\neq\!\emptyset$ be arbitrary.
Since the coupling~\eqref{eq:coupling} admits a solution $(\tilde{z}_i,\tilde{z}_{j})$ with $A_{j}^{(\tilde{z}_{j})}\!\cap\!A_{j}^\rat(h)\!\neq\!\emptyset$, at which the rational opponent has a feasible action to play, the joint outcome $(\Opol_i\,\|\,\Opol_{j})(h)$ is non-empty. Let $(a_i,a_{j})\!\in\!(\Opol_i\,\|\,\Opol_{j})(h)$ be any realized joint action. By~\eqref{eq:opinion_joint}, the induced opinions $\tilde{z}_i\!=\!\O_i(a_i)$ and $\tilde{z}_{j}\!=\!\O_{j}(a_{j})$ solve the coupling~\eqref{eq:coupling}, with $a_i\!=\!\Opol_i(h,\tilde{z}_i)$ and $a_{j}\!=\!\Opol_{j}(h,\tilde{z}_{j})$. Since $\Opol_{j}$ is rational, $a_{j}\!\in\!A_{j}^{(\tilde{z}_{j})}\!\cap\!A_{j}^\rat(h)$, so this coupling solution has $A_{j}^{(\tilde{z}_{j})}\!\cap\!A_{j}^\rat(h)\!\neq\!\emptyset$ and the hypothesis~\eqref{eq:coord2_cond} applies, giving $A_i^\rob(h;\tilde{z}_{j}) \cap A_i^{(\tilde{z}_i)}\!\neq\!\emptyset$, and the construction~\eqref{eq:coord2_construct} sets $a_i\!=\!\Opol_i(h,\tilde{z}_i)\!\in\!A_i^\rob(h;\tilde{z}_{j})$. By~\eqref{eq:rob_partition}, $a_i$ coordinates with $a_{j}$, hence $(a_i,a_{j})\!\in\!\pi(h)$. As $\Opol_{j}$, $h$, and the realized pair were arbitrary, $\Opol_i$ is $(\pi,\Sigma_{j}^{\text{op},\rat})$-robust.
\end{proof}

Compared with Theorem~\ref{th:op_coord_1}, Theorem~\ref{th:op_coord_2} imposes a more restrictive coupling between the two agents' opinions to ensure their mutual compatibility. Intuitively, when the opponent commits to a fixed action, agent $i$ only needs to adapt to it; when both agents react to each other's opinions, their reactions must be mutually compatible for the coordination prescribed by $\pi$. Under this condition, a $(\pi,\Sigma_{j}^{\text{op},\rat})$-robust strategy can be designed as illustrated in the proof.

In fact, the two theorems can be linked because a non-reactive rational opponent is a limiting case of an opinion-guided one. When $\Gamma_j(h)\!=\!\mathbf{0}$, agent $j$'s opinion state $\tilde z_j$ is decoupled from agent $i$'s opinion $\tilde z_i$ and stays fixed, so agent $j$ commits to a fixed action, i.e., a non-reactive opponent. Moreover, since real agents carry varying degrees of prior preference, encoded in the bias $b_{j}$ of the opinion dynamics~\eqref{eq:layered_opinion}, a single opinion-guided model spans the range from a fully reactive to a fully committed opponent. A strategy satisfying Theorem~\ref{th:op_coord_2} therefore achieves robust coordination against both reactive and non-reactive opponents, as formalized in the following corollary.



\begin{corollary}[Robust coordination against both reactive and non-reactive opponents]\label{cor:op_type_agnostic}
For a two-agent game $\mathcal{G}$ with permissive joint strategy $\pi$, suppose the partition functions $\O_i,\O_{j}$ and the other agent's opinion-pairing map $\Gamma_j$ are given. Suppose there exists an opinion-pairing map $\Gamma_i$ such that, for every $h\!\in\!\mathcal{H}$ with $\pi(h)\!\neq\!\emptyset$ and every opinion state $\tilde{z}_{j}\!\in\!\mathbb{B}^{n_j}$, the coupling~\eqref{eq:coupling} admits a solution $(\tilde{z}_i,\tilde{z}_{j})$, and every such solution satisfies~\eqref{eq:coord2_cond}. Then the opinion-guided strategy $\Opol_i$ constructed from $\Gamma_i$ via~\eqref{eq:coord2_construct} achieves robust coordination not only against reactive opinion-guided opponents $\Opol_j\!\in\!\Sigma_{j}^{\text{op},\rat}$ but also against non-reactive rational opponents $\Opol_j\!\in\!\Sigma_{j}^\rat$.
\end{corollary}
\begin{proof}
By assumption $\Gamma_i$ satisfies the conditions of Theorem~\ref{th:op_coord_2}, so $\Opol_i$ is $(\pi,\Sigma_{j}^{\text{op},\rat})$-robust. Because the coupling~\eqref{eq:coupling} is solvable at every $\tilde{z}_{j}\!\in\!\mathbb{B}^{n_j}$ and~\eqref{eq:coord2_cond} holds at every solution, the pairing $\tilde{z}_i\!=\!\Gamma_i(h)\tilde{z}_{j}$ satisfies~\eqref{eq:coord2_cond} for every $\tilde{z}_{j}$, which is exactly the full-rational condition~\eqref{eq:coordfull_cond} of Theorem~\ref{th:op_coord_full}. Hence $\Gamma_i$ also satisfies Theorem~\ref{th:op_coord_full}, so the same $\Opol_i$ is $(\pi,\Sigma_{j}^\rat)$-robust.
\end{proof}

\begin{tcolorbox}[runningexample]
\textbf{Running example}: When both robots run the strategy~\eqref{eq:example_strategy}, so that $\Gamma_i(s_0)\!=\!\Gamma_j(s_0)\!=\!-1$, the coupling~\eqref{eq:coupling} has two solutions, $(\tilde{z}_i, \tilde{z}_{j})\!=\!(+1,-1)$ and $(-1,+1)$. At both solutions, condition~\eqref{eq:coord2_cond} holds by the robust-set computation above, so Theorem~\ref{th:op_coord_2} guarantees that every joint action realized at $s_0$ lies in $\pi(s_0)$. Since the coupling is solvable at every $\tilde{z}_{j}$, Corollary~\ref{cor:op_type_agnostic} extends this guarantee to a non-reactive rational opponent as well, so the strategy is robust against any rational opponent, reactive or not. Finally, since the game is homogeneous, the strategy is also $\pi$-homogeneously implementable: both robots can run it and still realize only pairs in $\pi(s_0)$.
\end{tcolorbox}

Theorems~\ref{th:op_coord_1} and~\ref{th:op_coord_2} are stated for predefined $(\pi,\O_i, \O_{j}, \Gamma_j)$ and allow the two agents to be heterogeneous. When both agents are homogeneous, they run the same $\Opol$, and share a single partition $\O_i\!=\!\O_j\!=\!\O$ and a single opinion-pairing map $\Gamma_i(h)\!=\!\Gamma_j(h)\!=\!\Gamma(h)$ with identical opinion dimension $n_1\!=\!n_2$. Expanding Eq.~\eqref{eq:coupling} and eliminating $\tilde{z}_j$ gives
\[
\tilde{z}_i\!=\!\Gamma_i(h)\tilde{z}_j,\quad \tilde{z}_j\!=\!\Gamma_j(h)\tilde{z}_i \;\Longrightarrow\; \tilde{z}_i\!=\!\Gamma_i(h)\Gamma_j(h)\,\tilde{z}_i.
\]
Whatever opinion state the other agent holds, the two couplings must admit a matched pair $(\tilde{z}_i,\tilde{z}_{j})$, as Corollary~\ref{cor:op_type_agnostic} requires. For the shared map $\Gamma(h)$, such a pair exists at every $\tilde{z}_{j}\!\in\!\mathbb{B}^n$ when $\Gamma(h)^2$ acts as the identity on the components that $\Gamma(h)$ couples: an all-zero row leaves its component of $\tilde{z}_i$ free and constrains nothing. A nonzero entry in every row couples them all, and the condition is then involutivity, $\Gamma(h)^2\!=\!I$ (equivalently, $\Gamma(h)\!=\!\Gamma(h)^{-1}$), at every $h$.

\begin{corollary}[Robust coordination among homogeneous agents]\label{cor:op_homo}
Let $\mathcal{G}$ be a homogeneous game (Def.~\ref{def:symmetric}) with permissive joint strategy $\pi$, and suppose a shared partition $\O$ and an involutive pairing map $\Gamma$ (i.e., $\Gamma(h)\!=\!\Gamma(h)^{-1}$ for all $h$) satisfy the robust-coordination condition~\eqref{eq:coord2_cond}. Then the $(\pi,\Sigma_j^{\text{op},\rat})$-robust strategy $\Opol\!\in\!\Sigma^{\text{op},\rat}$ constructed from~\eqref{eq:coord2_construct} is also $\pi$-homogeneously implementable (Def.~\ref{def:robust_coord_homo}):
\begin{equation}\label{eq:op_homo}
\emptyset \neq (\Opol \,\|\, \Opol\!\circ\!\bar\tau)(h)\!\subseteq\!\pi(h), \quad \forall h\!\in\!\mathcal{H} \text{ with } \pi(h)\!\neq\!\emptyset.
\end{equation}
\end{corollary}
\begin{proof}
Write $\O\!:=\!\O_1\!=\!\O_2$ and $\Gamma\!:=\!\Gamma_1\!=\!\Gamma_2$ for the shared partition and pairing map, with $n_1\!=\!n_2\!=:\!n$. The involutive $\Gamma$ makes the coupling~\eqref{eq:coupling} solvable for every $\tilde{z}_{j}\!\in\!\mathbb{B}^n$, since $\tilde{z}_i\!=\!\Gamma(h)\tilde{z}_{j}$ gives $\Gamma(h)\tilde{z}_i\!=\!\Gamma(h)^2\tilde{z}_{j}\!=\!\tilde{z}_{j}$. Together with the assumed condition~\eqref{eq:coord2_cond}, $\Gamma$ satisfies the requirements of Theorem~\ref{th:op_coord_2}, so the strategy $\Opol$ built from it via~\eqref{eq:coord2_construct} is $(\pi,\Sigma_{j}^{\text{op},\rat})$-robust. By the definition of the robust action set~\eqref{eq:rob_partition}, a robust action pairs with the opponent's rational action to form a joint action in $\pi(h)$, and is therefore itself rational (Def.~\ref{def:rational}); hence the robust $\Opol$ must be rational for agent~$i$, i.e., $\Opol\!\in\!\Sigma_i^{\text{op},\rat}$.

The relabeling $\bar\tau$ is a symmetry of the homogeneous game, so the relabeled copy $\Opol\!\circ\!\bar\tau$ plays agent $j$ exactly as $\Opol$ plays agent $i$. Since $\Opol$ is rational for agent~$i$, this symmetry makes $\Opol\!\circ\!\bar\tau$ rational for agent~$j$: every action it selects lies in $A_j^\rat(h)$ at every history $h$. Hence $\Opol\!\circ\!\bar\tau$ is a rational opinion-guided strategy for agent $j$, i.e., $\Opol\!\circ\!\bar\tau \in \Sigma_{j}^{\text{op},\rat}$. Consequently, for every $h\!\in\!\mathcal{H}$ with $\pi(h)\!\neq\!\emptyset$, the joint outcome $(\Opol\,\|\,\Opol\!\circ\!\bar\tau)(h)$ is non-empty, and the robustness of $\Opol$ against $\Sigma_{j}^{\text{op},\rat}$ places every realized pair in $\pi(h)$. This is~\eqref{eq:op_homo}, so $\Opol$ is $\pi$-homogeneously implementable.
\end{proof}

Corollary~\ref{cor:op_homo} shows that in a homogeneous game, once a shared partition $\O$ and an involutive pairing map $\Gamma$ satisfy the robust-coordination condition~\eqref{eq:coord2_cond}, the resulting robust opinion-guided strategy is automatically homogeneously implementable. Combining this with Corollary~\ref{cor:op_type_agnostic}, a strategy built on the shared $(\O,\Gamma)$ achieves the coordination prescribed by $\pi$ against any rational opponent, whether reactive or non-reactive, and whether the opponent runs the identical or a different strategy.

In summary, this section establishes, step by step, the conditions under which an opinion-guided strategy can resolve the decentralized coordination challenges. The triple $(\pi, \O, \Gamma)$ plays a central role in these conditions. Given a scenario, $\pi$ specifies which joint actions count as successful coordination; $\O$ partitions each agent's (possibly infinite) action set into finitely many cells, each indexed by a discrete opinion state; and $\Gamma$ specifies how an agent sets its own opinion in reaction to the opponent's opinion. Intuitively, these three objects represent common knowledge shared by agents, analogous to the traffic conventions that drivers learn from experience.

\subsection{Convergence Analysis of Opinion Dynamics}

The decision-layer analysis above establishes robust coordination, provided the two agents settle on a matched pair of opinion states that solves the coupling~\eqref{eq:opinion_coupling}. However, it does not explain how the agents arrive at such a pair. We now show how this discrete coupling is realized by the continuous opinion dynamics in~\eqref{eq:layered_opinion}.

Under the linear-pairing Assumption~\ref{asm:linear_pairing} on $\Gamma_i(h)$, each component of agent $i$'s opinion state couples to at most one component of agent $j$'s. The full opinion dynamics~\eqref{eq:layered_opinion} therefore separate into independent scalar subsystems, and their convergence analysis reduces to that of a single scalar subsystem. Each subsystem inherits a single coupling entry $\gamma\!\in\!\{+1,0,-1\}$ from $\Gamma_i(h)$.

For analytical simplicity, we further consider agents with no prior preference, $b_1\!=\!b_2\!=\!0$, set the self- and opponent-coupling gains to a common value $\kappa\!>\!0$, and take $\rho_1\!=\!\rho_2\!=\!\rho$. Under these assumptions, the two-agent dynamics for scalar opinions $(z_1, z_2)$ are given by
\begin{equation}\label{eq:NOD}
\left\{ \! \begin{array}{l} \textstyle
\dot{z}_1\!=\!-d z_1 + \rho \tanh\!\left(\kappa(z_1\!+\!\gamma z_2)\right) \\
\dot{z}_2\!=\!-d z_2 + \rho \tanh\!\left(\kappa(z_2\!+\!\gamma z_1)\right)
\end{array}\right.,
\end{equation}
whose adjacency matrix is $A\!=\!\textstyle\begin{bmatrix}0 & \gamma\\ \gamma & 0\end{bmatrix}$.

If $\gamma\!=\!0$, this component is decoupled from the opponent's opinion. Starting from the neutral initial value $z_i\!=\!0$, it remains at $z_i\!=\!0$, leaving the agent free to assign either discrete opinion state to it regardless of the opponent's choice. When $\gamma\!=\!\pm1$, we analyze how the two opinions evolve and converge to a matched pair of discrete opinion states.  By~\cite[Cor.~IV.1.2]{bizyaeva2022nonlinear} applied to system~\eqref{eq:NOD}, the neutral opinion $\boldsymbol{z}\!=\![z_1\ z_2]^\top\!=\!\mathbf{0}$ is locally exponentially stable for $\rho\!<\!\rho^\star$ and loses stability for $\rho\!>\!\rho^\star$, with critical value $\rho^\star\!=\!\frac{d}{2\kappa}$. At $\rho\!=\!\rho^\star$, the system undergoes a supercritical pitchfork bifurcation: two branches emerge from the origin, tangent to $\boldsymbol{v}_{\max}\!=\![1\ \gamma]^\top$, the eigenvector associated with the largest eigenvalue of $A$ (Fig.~\ref{fig:bifur}). To activate the bifurcation for decisive coordination, we set
\begin{equation}
\rho = \frac{d}{2\kappa} + \epsilon
\label{eq:rho_set}
\end{equation}
with $\epsilon\!>\!0$ small, placing $\rho$ just past $\rho^\star$. The opinion state $\boldsymbol{z}$ then converges to one of the two stable equilibria $\boldsymbol{z}^+$ or $\boldsymbol{z}^-$, small positive and negative multiples of $\boldsymbol{v}_{\max}$ respectively, corresponding to a matched pair of discrete opinion states $(\tilde{z}_1,\tilde{z}_2)$ with $\tilde{z}_1\!=\!\gamma\,\tilde{z}_2\!\in\!\{+1,-1\}$, the role assignment prescribed by $\gamma$. The convergence time can be reduced by increasing $\epsilon$, allowing the opinion dynamics to settle on such a matched pair, after which the strategy layer commits to the corresponding joint action in $\pi(h)$.

\begin{figure}[t]
    \centering
    \subfigure[$\gamma\!=\!+1$]{\label{fig:bifur_pos}
        \includegraphics[width=0.465\linewidth]{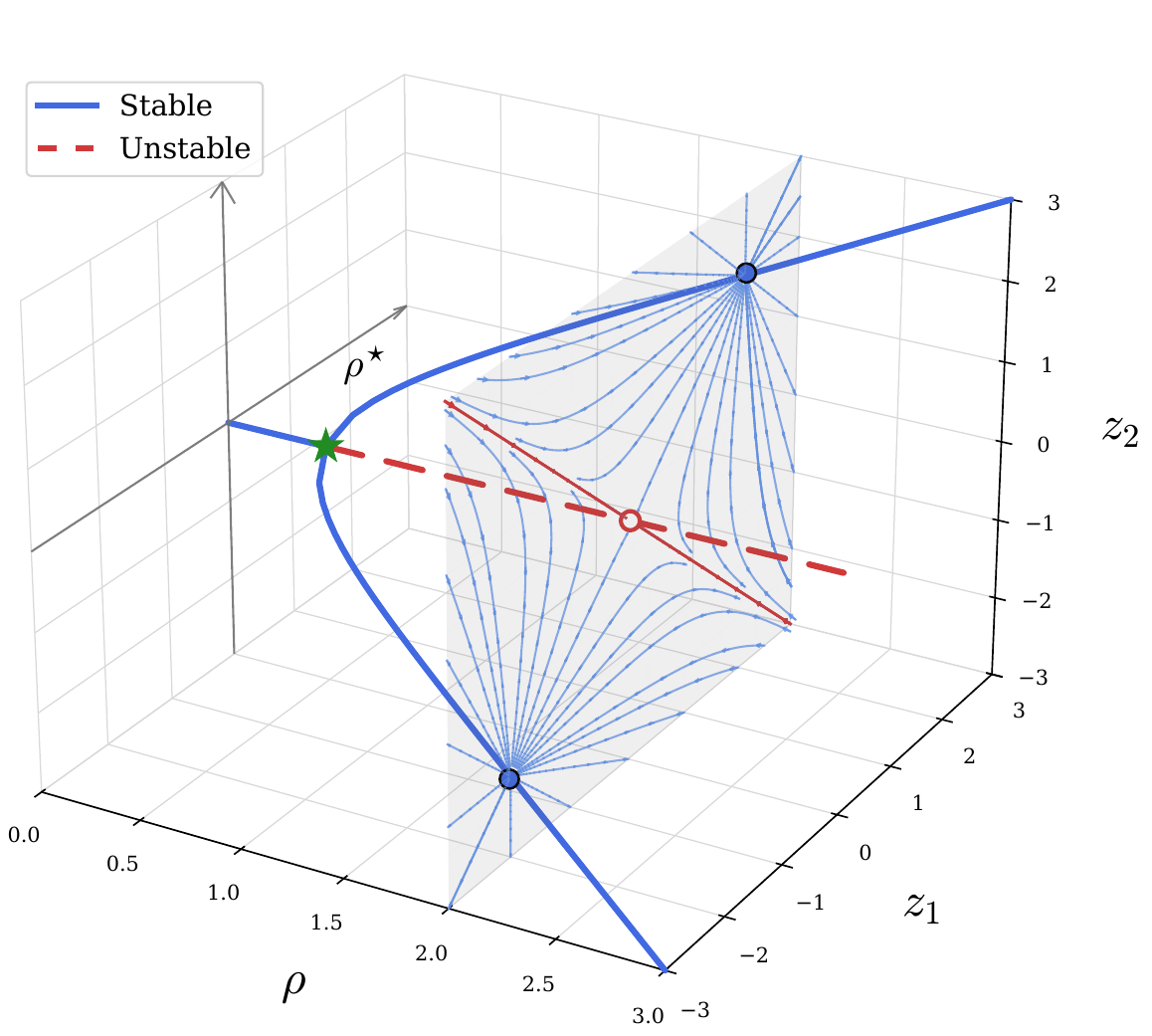}
    }
    \subfigure[$\gamma\!=\!-1$]{\label{fig:bifur_neg}
        \includegraphics[width=0.465\linewidth]{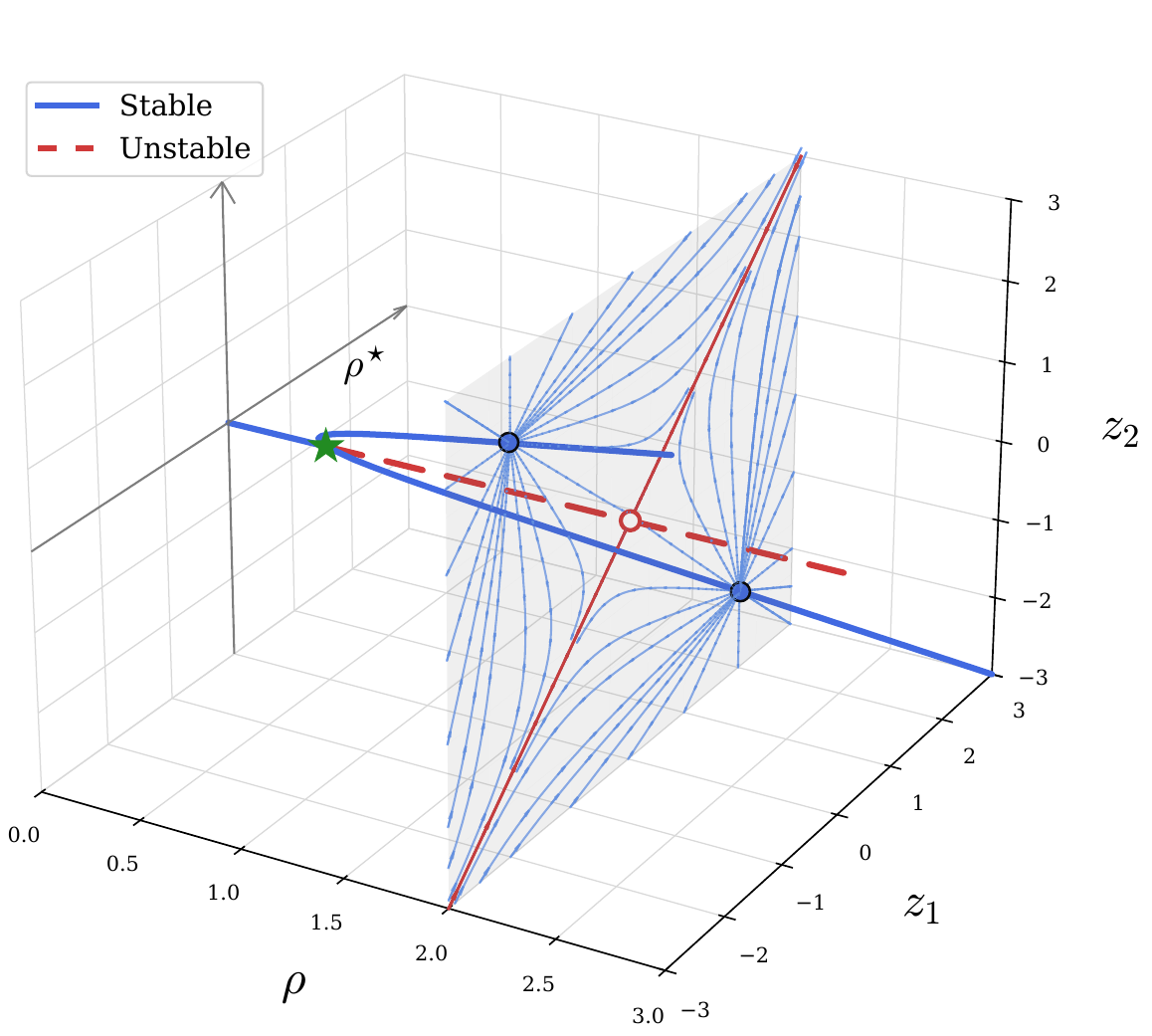}
    }
    \caption{Bifurcation diagrams of the two-agent scalar opinion dynamics~\eqref{eq:NOD} as the attention gain $\rho$ increases, shown for the two coupling signs $\gamma$. \textbf{(a)}~$\gamma\!=\!+1$: the two stable branches emerge along $\pm[1,1]^\top$, pairing same-sign opinions. \textbf{(b)}~$\gamma\!=\!-1$: the stable branches emerge along $\pm[1,-1]^\top$, pairing opposite-sign opinions, as in the running example~\eqref{eq:od_running}. In both panels the neutral equilibrium $\boldsymbol{z}\!=\!\mathbf{0}$ is stable for $\rho\!<\!\rho^\star$ and undergoes a supercritical pitchfork at the threshold $\rho^\star\!=\!d/(2\kappa)\!=\!0.5$ (green star); solid blue curves are the stable equilibrium branches and the dashed red curve is the unstable branch. The vertical slice $\rho\!=\!2$ carries the phase portrait, whose trajectories diverge from the unstable equilibrium (open circle) and converge to the two stable equilibria (filled circles), with the red curves marking the separatrix.}
    \label{fig:bifur}
\end{figure}

\subsection{Layered Guarantee}
Given the above analysis of the two layers, a guarantee of robust coordination with respect to $\pi$ can be obtained for the whole layered opinion-guided framework, based on a single time-scale separation between them:
\begin{assumption}[Time-scale separation]\label{asm:timescale}
The strategy layer acts at times $t_k$, $k\!\in\!\mathbb{N}$. Over each period $[t_k,t_{k+1})$:
\begin{enumerate}[leftmargin=*, align=left]
\item the history $h$ of $\mathcal{G}$ remains constant;
\item the window $\Delta t_{\mathrm{op}}\!<\! t_{k+1} - t_k$ is long enough that the opinion state $\tilde{z}_i\!=\!\mathrm{sgn}(z_i)$ remains constant over $[t_k\!+\!\Delta t_{\mathrm{op}},t_{k+1})$.
\end{enumerate}
\end{assumption}

With this assumption, the strategy layer selects a joint action at each decision time $t_k$ based on the opinion states at that time, and the opinion dynamics have enough time to settle on a matched pair of opinion states before the next decision time. The following theorem then establishes a layered guarantee for robust coordination.

\begin{theorem}[Layered guarantee for opinion-guided coordination]\label{th:opinion_robust_coord}
Consider a two-agent game $\mathcal{G}$ with permissive joint strategy $\pi$, where both agents share the partition map $\O$ and an opinion-pairing map $\Gamma$ satisfying Assumption~\ref{asm:linear_pairing}. If $\Gamma$ satisfies the conditions of Corollary~\ref{cor:op_type_agnostic}, the attention gain obeys $\rho\!>\!\rho^\star\!=\!\tfrac{d}{2\kappa}$, and the time-scale separation of Assumption~\ref{asm:timescale} holds, then the layered opinion-guided strategy $\Opol_i$ built from $\Gamma$ via~\eqref{eq:coord2_construct} achieves robust coordination with respect to $\pi$ against every $\pi$-rational opponent, whether reactive or non-reactive, and whether or not it runs the identical strategy as $\Opol_i$.
\end{theorem}
\begin{proof}
By Corollary~\ref{cor:op_type_agnostic}, $\Opol_i$ is robust against every $\pi$-rational opponent, reactive or non-reactive: any joint action realized at a history $h$ with $\pi(h)\!\neq\!\emptyset$ lies in $\pi(h)$. It thus remains to show that the opinion layer realizes one within each period. With $\rho\!>\!\rho^\star\!=\!\tfrac{d}{2\kappa}$, the neutral opinion $\boldsymbol{z}\!=\!\mathbf{0}$ of~\eqref{eq:NOD} is unstable, so from almost every initial condition the opinions converge to a stable equilibrium aligned with $\pm\boldsymbol{v}_{\max}\!=\!\pm[1\ \gamma]^\top$, whose signs form a matched pair $\tilde{z}_i\!=\!\gamma\tilde{z}_{j}$ that solves the coupling~\eqref{eq:coupling} (against a non-reactive opponent, its committed action fixes $\tilde{z}_{j}$ and $z_i$ converges to the matched sign). By Assumption~\ref{asm:timescale}, these signs settle within the reserved window and the action executes before the next decision time, so a joint action is realized at $h$; by robustness, $(a_i,a_{j})\!\in\!\pi(h)$.
\end{proof}

\section{Practical Implementation and Case Studies} \label{sec:case_study}

\begin{figure*}[t]
    \centering
    \includegraphics[width=\textwidth]{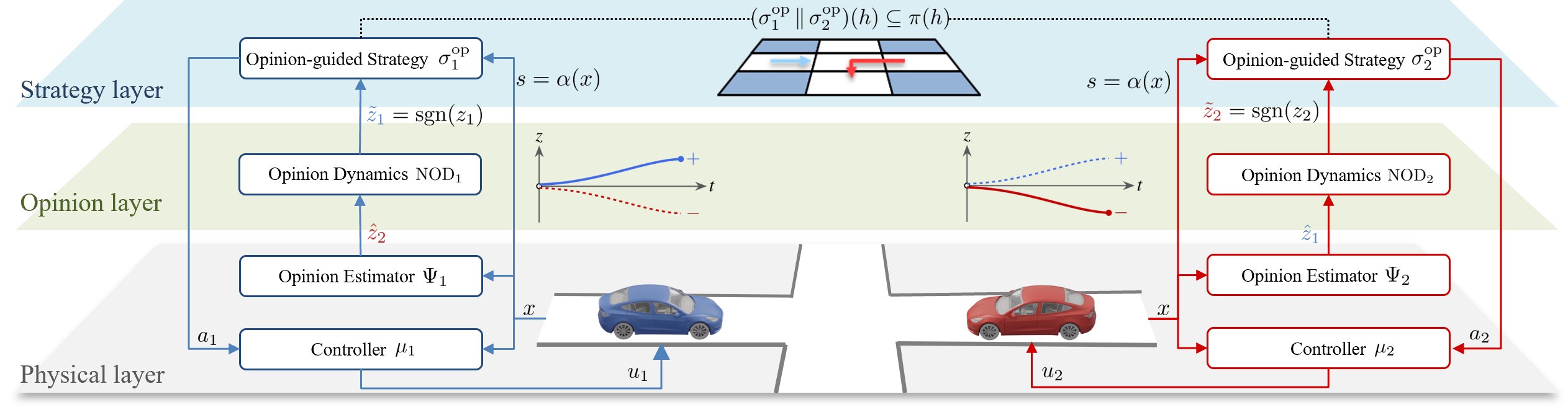}%
    \caption{Block-diagram view of the communication-free opinion-guided coordination framework. Each agent (agent~1 in blue, agent~2 in red) runs an opinion-guided strategy ($\Opol_i$) in the strategy layer, opinion dynamics (NOD$_i$) in the opinion layer, and an opinion estimator ($\Psi_i$) together with a tracking controller ($\mu_i$) acting on the physical plant ($f_i$) in the physical layer. The two control stacks are coupled only through the jointly observable physical state $x\!=\!(x_1,x_2)$: each opinion estimator $\Psi_i$ infers the opponent's opinion $\hat z_{j}$ from $x$ and feeds it into NOD$_i$, so the opinion loop closes through shared observation rather than direct inter-agent communication.}
    \label{fig:opinion-block}
\end{figure*}

\subsection{Communication-Free Implementation}\label{subsec:comm_free} The opponent-coupling term in~\eqref{eq:layered_opinion} requires agent $i$ to read its opponent's opinion state $z_{j}$ directly, which presumes inter-agent communication. In practice, such communication is often impractical or unavailable. Yet human drivers coordinate without it, because each agent's continuous motion already discloses the action it is committing to. The opinion-guided strategy admits a communication-free implementation, depicted in Fig.~\ref{fig:opinion-block}.

The communication-free implementation extends the layered opinion-guided strategy in~\eqref{eq:layered_opinion} by adding a physical layer. In the physical layer, agent $i$ has a physical state $x_i$ governed by the plant dynamics $\dot{x}_i\!=\!f_i(x_i,u_i)$, and a low-level tracking controller $\mu_i$ maps the committed action $a_i$ to an input $u_i\!=\!\mu_i(x,a_i)$ that steers $x_i$ along the motion realizing $a_i$. The two agents share only the jointly observable state $x\!=\!(x_1,x_2)\!\in\!\mathcal{X}$. From $x$, each opinion estimator obtains an estimate $\hat{z}_{j}\!=\!\Psi_i(x,h)$ of the opponent's opinion. In practice, the estimator $\Psi_i$ can be implemented with a variety of techniques, such as Kalman filtering, particle filtering, or learning-based methods. 

Unlike the realization in~\eqref{eq:layered_opinion}, the communication-free implementation replaces the directly read opponent's opinion $z_{j}$ with the estimate $\hat{z}_{j}$. The loop therefore closes without any inter-agent communication: the strategy layer commits an action from the current opinion state $\tilde{z}_i$ and the history $h$ assembled from the abstracted states $s\!=\!\alpha(x)$ via a state abstraction map $\alpha\!:\!\mathcal{X}\!\to\!\mathcal{S}$; the tracking controller drives the plant to realize it, so the agent's own motion discloses that action through the shared state $x$; the opinion estimator infers the opponent's opinion $\hat{z}_{j}\!=\!\Psi_i(x,h)$ from $x$; and the opinion dynamics evolve under this estimated coupling signal, feeding the resulting opinion state back to the strategy layer for the next decision. The coordination guarantees of Section~\ref{sec:analysis} hold if $\hat{z}_{j}$ is accurate enough for the opinion dynamics to converge to the same matched pair of signs. Owing to the fast and robust bifurcation of the opinion dynamics, the two agents are driven to achieve a pair of actions in the permissive set $\pi(h)$.

To deploy the proposed opinion-guided layered control framework in a specific application, the following components must be instantiated. (i) Strategy layer: from the desired joint task and its specifications, first construct a permissive joint strategy $\pi$ collecting every joint action pair that accomplishes the task, which in turn determines the rational strategy set $\Sigma_j^\rat$ of agent $j$; then design each agent's strategy $\Opol_i$ from a partition $\O$ of the opponent's action set and an opinion-pairing map $\Gamma$ that together satisfy the conditions of Corollary~\ref{cor:op_type_agnostic}, so that $\Opol_i$ coordinates robustly with respect to $\pi$ and $\Sigma_j^\rat$. (ii) Opinion layer: opinion dynamics that undergo the intended bifurcation and converge quickly and robustly, as required by the time-scale separation of Assumption~\ref{asm:timescale}. (iii) Physical layer, which realizes the committed actions on the plant: a state abstraction map $\alpha$ that recovers the discrete game state from the shared physical state, a low-level controller $\mu_i$ that tracks and executes the committed action, and an opinion estimator $\Psi_i$ that infers the opponent's opinion from the shared physical observation. 

To demonstrate how the layered framework coordinates decentralized agents, the three case studies below ground it in settings of growing complexity: Case-1 resolves a single conflict over a continuous action set with a scalar opinion; Case-2 spans multiple task stages with a multi-dimensional opinion; and Case-3 extends to a sequential general-sum game played over many decision steps.

\subsection{Case-1: On-Ramp Merging}

In the single-lane ramp merging scenario as shown in Fig.~\ref{fig:ramp_merge}, the ego vehicle faces an opponent vehicle whose behavior is uncertain: the opponent may either yield or maintain its speed, which determines whether the ego vehicle should merge in front of it or behind it. For simplicity, we consider only longitudinal dynamics and model each agent $i\!\in\!\{e,o\}$ as a single integrator $\dot{p}_i\!=\!v_i$, where $p_i$ is the longitudinal position along its lane and the speed $v_i\!\in\![v_i^{\min}, v_i^{\max}]$ is the control input.

In this case, the joint state is the pair of vehicle positions $s\!=\!(p_e,p_o)$, which coincides with the physical state; the abstraction map $\alpha$ is therefore the identity, and no state abstraction is required.  Each agent's action is a constant speed held over one decision interval, and its action set is the continuous feasible range $A_i\!=\![v_i^{\min}, v_i^{\max}]$, $i\!\in\!\{e,o\}$. We take both the permissive strategy $\pi$ and each agent's individual strategy to be memoryless (Markov), so every decision depends only on the current state $s$. The permissive set $\pi(s)\!\subseteq\!A_e\!\times\!A_o$ is determined by the merge geometry. The time for vehicle $i$ to reach its merge point $p_i^{\mathrm{m}}$ at constant speed is $\tau_i\!=\!(p_i^{\mathrm{m}}\!-\!p_i)/v_i$. The joint speed $(v_e, v_o)$ is permissive when the following vehicle, on reaching the merge point, trails the leader by at least a safety margin $d_{\text{safe}}$,
\begin{equation}\label{eq:merge_permissive}
\begin{aligned}
\pi(s) \!=\! \big\{ (v_e,v_o)\!\in\!A_e\!\times\!A_o \mid\ & v_e(\tau_o\!-\!\tau_e)\!\ge\!d_{\text{safe}} \\[-1pt]
\text{or}\ \ & v_o(\tau_e\!-\!\tau_o)\!\ge\!d_{\text{safe}} \big\}.
\end{aligned}
\end{equation}
The first inequality gives a front merge for the ego vehicle, and the second a rear merge. Fig.~\ref{fig:ramp_partition} draws $\pi(s)$ at a representative state as the green region. The rational action set $A_i^\rat(s)$ of each agent is the projection of $\pi(s)$ onto that agent's speed axis, drawn as the light-colored line along each axis.

To design the opinion-guided strategy, we first let the partition function $\O_i$ split each agent's action set at its midpoint $v_i^{\mathrm{th}}\!=\!\frac{1}{2}(v_i^{\min} + v_i^{\max})$, $i\!\in\!\{e,o\}$, into the two cells $A_{i}^{(-)} \!=\! [v_i^{\min},\, v_i^{\mathrm{th}})$ and $A_{i}^{(+)} \!=\! [v_i^{\mathrm{th}},\, v_i^{\max}]$. Merging is an anti-coordination task: the two vehicles resolve the conflict by taking opposite roles, one accelerating to pass through the merge point first and the other yielding to merge behind. We encode this structure with the constant opinion-pairing map $\Gamma\!=\!-1$, i.e.\ $\gamma\!=\!-1$ in~\eqref{eq:NOD}. Evaluating the robust action set~\eqref{eq:rob_partition} on the inferred opponent opinion gives $A_{e}^\rob(s;\tilde{z}_{o})$, the ego speeds that remain permissive against every rational opponent speed in that cell, shown in Fig.~\ref{fig:ramp_partition}. The ego estimates the opponent's speed by differencing the observed position over a window $\delta$, $\bar{v}_{o}\!=\![p_o(t)\!-\!p_o(t\!-\!\delta)]/\delta$, and thresholds it to recover the opponent's opinion state, $\hat{z}_{o}\!=\!\mathrm{sgn}(\bar{v}_{o}\!-\!v_o^{\mathrm{th}})\!=\!\mathrm{sgn}(p_o(t)\!-\!p_o(t\!-\!\delta)\!-\!v_o^{\mathrm{th}}\delta)$. Driven by this estimate, the opinion layer runs the scalar dynamics~\eqref{eq:NOD} with $\rho$ past the bifurcation threshold~\eqref{eq:rho_set}, converging to a matched opinion state that commits the ego to $A_{e}^\rob(s;\tilde{z}_{o})$. The ego then takes the speed within it closest to its reference $v_e^{\mathrm{ref}}$,
\begin{equation}\label{eq:merge_select}
v_e \;=\; \operatorname*{arg\,min}_{v \in A_{e}^\rob(s;\tilde{z}_{o})} \, \lvert v - v_e^{\mathrm{ref}} \rvert,
\end{equation}
so it stays near the reference speed whenever coordination does not force it away.

To verify that the designed strategy is effective across the range of rational opponents established in Corollary~\ref{cor:op_type_agnostic}, we evaluate it against two families of opponents. The first is a non-reactive individual strategy that commits to a fixed speed sampled uniformly at random from $A_o$. The second runs the same opinion-guided strategy as the ego but with a randomly sampled prior bias $b_o$; as noted in Section~\ref{sec:analysis}, this bias spans the range from a fully reactive to an effectively committed opponent, so the two families together cover the range of Corollary~\ref{cor:op_type_agnostic}. Over $100$ trials in each setting, the ego completes a safe merge in every trial, keeping the joint speed within the permissive set and the merge separation above $d_{\text{safe}}$, and it commits to the role compatible with the opponent: against a fixed-speed strategy it adapts to the committed role, and against a biased opinion-guided opponent it commits to the role opposite the one the opponent's bias drives it toward. To stress-test the strategy further, we consider a time-varying opponent that alternates between high and low speeds. As shown in Fig.~\ref{fig:ramp_merge}, the ego adaptively selects its speed from the robust action set of the opponent's current behavior; despite the rapid switching, it reacts promptly and keeps the joint state within the permissive set throughout (Fig.~\ref{fig:ramp_partition}).

\begin{figure}[t]
    \centering
    \includegraphics[width=\linewidth]{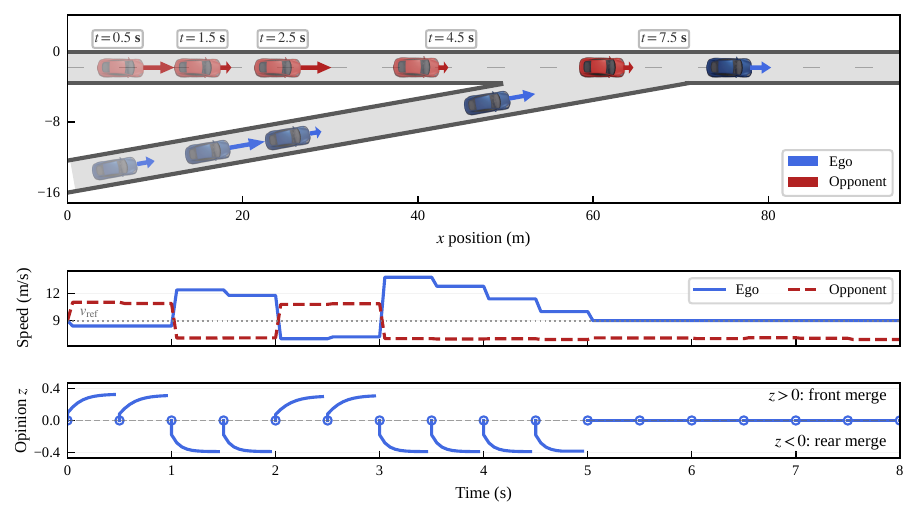}
    \caption{One-lane ramp merging against a speed-switching opponent. Top: road scene at successive time instants, where the ego (blue) merges into the lane of the fluctuating opponent (red). Middle: ego and opponent speeds, with the ego reference speed marked by the dotted line. Bottom: the ego opinion state $z$, whose sign commits the ego to a front ($z>0$) or rear ($z<0$) merge.}
    \label{fig:ramp_merge}
\end{figure}

\begin{figure}[htb]
    \centering
    \subfigure[$t=1.5$ s]{\label{fig:ramp_part_1}
        \includegraphics[width=0.46\linewidth]{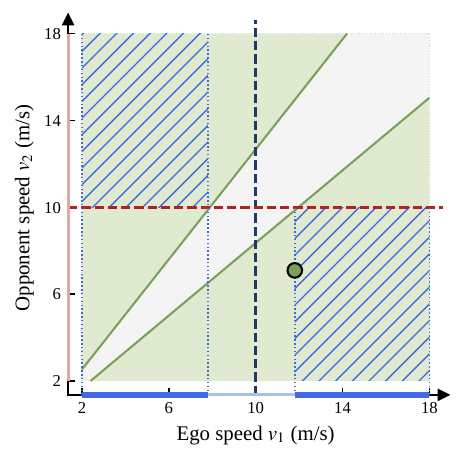}
    }
    \subfigure[$t=2.5$ s]{\label{fig:ramp_part_2}
        \includegraphics[width=0.46\linewidth]{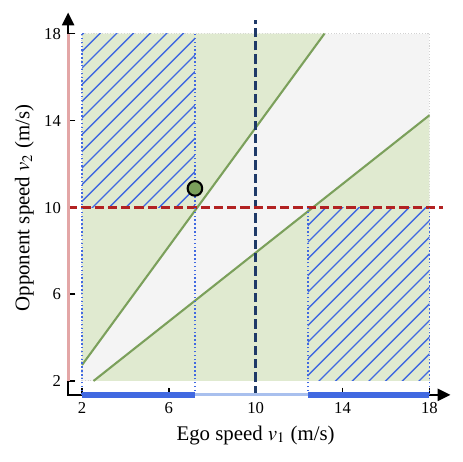}
    }
    \caption{Permissive and robust action sets in the joint speed space $A_e\!\times\!A_o$, at two instants of the run in Fig.~\ref{fig:ramp_merge}. The green region is the permissive set $\pi(s)$; the grey diagonal band is the unsafe region where the vehicles reach the merge point closer than $d_{\text{safe}}$. Dashed lines mark the partition thresholds $v_o^{\mathrm{th}}$ and $v_e^{\mathrm{th}}$; the light segments on the two axes are the agents' rational sets $A_e^\rat(s)$ and $A_o^\rat(s)$, while the dark segments on the ego axis are its robust action sets $A_e^\rob(s;\tilde{z}_{o})$. The green dot is the executed joint speed at this state.}
    \label{fig:ramp_partition}
\end{figure}

\subsection{Case-2: Decentralized Dual-Arm Manipulation}

\begin{figure*}[thb]
    \centering
    \includegraphics[width=0.97\textwidth, trim={55 0 0 0}, clip]{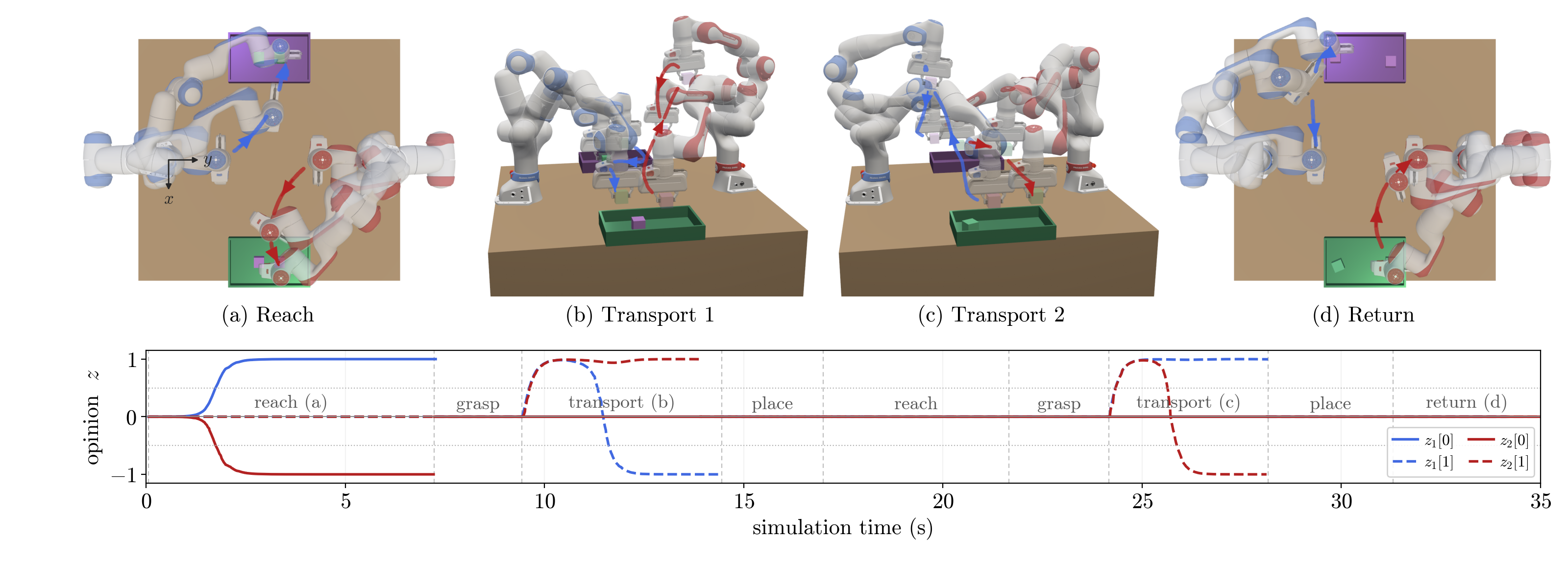}
    \caption{Decentralized dual-arm manipulation guided by the opinion state. Top: translucent snapshots of the two arms (arm~1 blue, arm~2 red) at successive instants along their trajectories. (a)~Reaching: the arms commit to opposite left/right halves of the table to grasp their cubes. (b)--(c)~Transporting: the arms carry the cubes through the shared central space toward the baskets, additionally separating along the upper/lower axis. (d)~Returning: the arms retract to their initial poses. Bottom: evolution of each arm's two opinions, with the solid curves the left/right opinion $z_i[0]$ and the dashed curves the upper/lower opinion $z_i[1]$.}
    \label{fig:dualarm}
\end{figure*}

As a second case study, we evaluate the proposed framework in a tabletop manipulation scenario, simulated in Gazebo~\cite{koenig2004gazebo}, in which two Franka Research~3 (FR3) arms~\cite{haddadin2022franka} operate concurrently within a shared workspace. The two arms are mounted on opposite sides of a table and controlled independently, as shown in Fig.~\ref{fig:dualarm}(a). Each of the two baskets starts with two cubes of the other basket's color, and the task is to sort them so that every cube ends in the basket matching its own color. Since an arm carries one cube at a time, it transports its two cubes in succession, and every transport crosses the whole table, so the two arms may enter the shared central space at the same time and collide or block each other. A natural way to coordinate the arms in this shared space is to keep them as far apart as possible, so that at any instant they occupy complementary parts of the workspace rather than the same region.

To make this precise, we fix a table frame common to both arms and partition the workspace along two independent axes of this frame. We write $c_i\!\in\!\mathbb{B}^2$ for the region arm $i$ occupies: the first entry encodes which lateral half of the table the arm lies in, left or right in this common frame, and the second whether it keeps to the upper or lower part of the shared central space, so that one arm passes above the other where the two cross. Each arm's action is the region it commits to next, so an action carries the same encoding as $c_i$ and $A_i\!=\!\mathbb{B}^2$. The partition $\O_i$ is then the identity and each cell $A_i^{(\tilde{z}_i)}$ holds a single action, in contrast to Case-1, where $\O_i$ collapses a continuum of speeds onto two cells. The permissive set $\pi(s)$ collects the joint actions in which the two arms hold complementary regions along the axis on which they must separate, and the committed region is realized in the physical layer, where a sampling-based planner generates a path confined to that region for the low-level controller $\mu_i$ to track. Each arm maintains an opinion $z_i\!\in\!\mathbb{R}^2$, a pair of scalar opinions, one per axis, whose signs form the opinion state $\tilde{z}_i\!=\!\mathrm{sgn}(z_i)\!\in\!\mathbb{B}^2$ that commits the arm to a region along the left/right and upper/lower axes.

Each arm's task unfolds in five stages: reach, grasp, transport, place, and return, with the first four repeated once per cube, so the arms complete two rounds before the final return. The arms coordinate through a shared game state: the abstraction $\alpha$ maps their configurations and grasp status to $s\!=\!(k, c_1, c_2)$, where $k$ is the current stage and $c_i$ is the region of the partitioned workspace that arm $i$'s end-effector lies in. Different stage requirements impose different separations on the two arms. While reaching, they must separate laterally to grasp from opposite halves of the table; while transporting, they must separate vertically to cross the shared central space without blocking each other; while grasping, placing, and returning, they need not separate at all, since their workspaces do not overlap. The arms therefore share one pairing map throughout, $\Gamma\!=\!-I$, and both run the opinion dynamics~\eqref{eq:layered_opinion} with a per-component attention $\rho_i\!=\!\mathrm{diag}(\rho_{i,1},\rho_{i,2})$ that decides which of its two separations is produced: a component with $\rho_{i,k}\!>\!\rho^\star$ bifurcates and the arms take opposite sides along it, while a component with $\rho_{i,k}\!<\!\rho^\star$ stays at the neutral opinion and no separation is made.


Fig.~\ref{fig:dualarm} shows how the opinion guides the two arms to coordinate in the shared workspace. In (a), the left/right opinion bifurcates and drives the arms to opposite lateral halves of the table to grasp their cubes, rather than reaching into the same half. In (b) and (c), the upper/lower opinion bifurcates and drives the arms to opposite halves of the shared central space, so that one passes above and the other below, and they cross without blocking each other. Notably, the arm passing above in the first crossing passes below in the second, so the upper/lower assignment is resolved anew at each crossing rather than fixed in advance. The opinion estimator recovers the other arm's opinion from the observed motion of its end-effector, so the two arms can coordinate without any communication. The estimator's design is not a contribution of this paper, so we omit the details for space.

We compare the proposed strategy against two baselines, running $20$ trials of each. The first removes the opinion guidance by fixing $\Gamma(s)$ to the zero matrix in every stage; with the two arms no longer coupled, they contend for the same region and block each other, and the task succeeds in only $3$ of the $20$ trials. The second replaces the opinion guidance with a fixed priority: one arm is the designated leader that always takes its preferred region while the other yields to an alternative. This breaks the symmetry, but does so rigidly: the follower must yield even when it is better placed to proceed, which wastes motion whenever the two arms prefer the same region, and stalls the task entirely if the leader is delayed, since the follower keeps waiting for an arm that cannot move; it succeeds in $11$ of the $20$ trials. The proposed strategy succeeds in all $20$.

\subsection{Case-3: Autonomous Navigation in Constrained Space}

In the first two case studies, the application itself dictates which joint behaviors are admissible, so the permissive strategy is constructed from an intuitive understanding of the practical task. In this case study, we instead ground the framework in a classic problem of game theory, a sequential general-sum game, where the admissible joint behaviors correspond to a well-defined solution concept, the Nash equilibria of the game. For such a well-defined game problem, we showcase how the proposed opinion-guided strategy and its layered realization resolve the long-standing equilibrium selection problem noted in the introduction, which arises whenever the game admits more than one equilibrium. 

\subsubsection{Navigation as a General-Sum Game} Specifically, we consider a two-agent navigation problem in a constrained environment (Fig.~\ref{fig:hutong}), the setting of the autonomous-vehicle deadlocks reported in recent news~\cite{YouTubeVideo_jam, YouTubeVideo_honk}. We formulate this problem as a sequential general-sum game as follows. Let $x_i\!\in\!X_i\!\subseteq\!\mathbb{R}^2$ denote the planar position of vehicle $i$, which serves as its physical state. Both vehicles operate in the same constrained environment, so $X_1\!=\!X_2\!=\!X$ and the joint physical state space is $\mathcal{X}\!=\!X\!\times\!X$. A straightforward abstraction is a grid, as illustrated in Fig.~\ref{fig:hutong_outcomes} and Fig.~\ref{fig:hutong_overall}: both vehicles share one map $\alpha_0\!:\!X\!\to\!S$ that sends a position to its grid state, and applying it to both gives the joint abstraction $\alpha\!:\!\mathcal{X}\!\to\!\mathcal{S}$ with $\mathcal{S}\!=\!S\!\times\!S$, which returns the pair of grid states $(s_1, s_2)$ of the two vehicles. For a two-vehicle navigation, the grid abstraction need not cover the whole environment: a $3\!\times\!3$ map suffices to capture the local interaction.

Each vehicle's action set consists of five semantic maneuvers, $A_i = \{\texttt{forward},\, \texttt{back},\, \texttt{left},\, \texttt{right},\, \texttt{wait}\}$, defined in a map frame common to both vehicles as the unit displacements $(0,1)$, $(0,-1)$, $(-1,0)$, $(1,0)$, and $(0,0)$; the transition relation $\rightarrow$ is thus deterministic: from each state $s$, each joint action $(a_1, a_2)$ leads to a unique successor $s'$, written $s \xrightarrow{(a_1, a_2)} s'$. Let $S_{\mathrm{obs}}\!\subset\!S$ denote the obstacle grid states, which no vehicle can drive through. The task fails if a vehicle enters an obstacle grid state or the two vehicles collide, and the states in which either has occurred form the failure set,
\begin{equation}\label{eq:sfail}
\mathcal{S}_{\mathrm{fail}} \!:=\! \bigl\{\, s\!\in\!\mathcal{S} \;\big|\; s_1 \!=\! s_2 \ \text{ or } \ \{s_1, s_2\} \cap S_{\mathrm{obs}} \!\neq\! \emptyset \,\bigr\}.
\end{equation}

\subsubsection{Nash Equilibrium Set as the Permissive Strategy} As in real-world deployments, each vehicle has a distinct navigation goal $g_i \!\in\! S$, assumed known to both vehicles. We carry it in the state: the augmented state of vehicle $i$ is $\bar{s}_i\!=\!(s_i, g_i)$, and the joint state is the pair $s\!=\!(\bar{s}_1, \bar{s}_2)$, on which each vehicle plays a standard Markov strategy $a_i\!=\!\sigma_i(s)$. We keep writing $\mathcal{S}$ for the set of such augmented joint states, and $s_1, s_2$ for the grid states of the two vehicles, as in~\eqref{eq:sfail}. Note that each vehicle's goal $g_i$ is constant throughout an interaction. Both vehicles share the same reward function $R$, each evaluated with respect to its own goal, which rewards goal attainment, penalizes collisions, and charges a constant step cost that discourages delay. Its full form is given in Appendix~\ref{app:reward}. We assume a shared $R$ to make the game homogeneous in its payoffs as well as its transitions (Def.~\ref{def:symmetric}), so that two strategies can differ only in how they select among equilibria, which keeps the later discussion clean. Since each vehicle optimizes the shared reward under its own goal, the two vehicles are neither fully cooperative nor purely adversarial: the game is general-sum. In a general-sum game, no single quantity ranks a joint maneuver as desirable, and all that can be required is that the two actions be best responses to each other, that is, a Nash equilibrium (NE), at which neither vehicle can improve its own return by changing its action alone. We therefore resolve the interaction by seeking NEs, as in other game-theoretic planning approaches for interactive driving~\cite{fisac2019hierarchical, li2018game, huang2024integrated, fridovich2020efficient}.

As noted above, general-sum games may admit multiple NEs. To retain this multiplicity, the permissive joint strategy $\pi$ is computed to collect every Nash action pair of the stage game at each state. Standard Nash value iteration resolves this multiplicity at synthesis, through a selection rule that is left arbitrary~\cite{kearns2000fast, hu2003nash}. Algorithm~\ref{alg:vi_ipc}, in which $\lambda\!\in\!(0,1)$ is the discount factor, departs from it in two ways. First, the algorithm returns a permissive strategy that retains the whole set of equilibria rather than committing to one of them. The operator $\textsc{Nash}(Q_1, Q_2)$ in Line~\ref{line:policy_extract} enumerates every pure-strategy Nash action pair of the stage game $(Q_1, Q_2)$ by exhaustive best-response checking over the $|A_1|\!\times\!|A_2|$ joint actions~\cite{bacsar1998dynamic, hespanha2017noncooperative}, and Line~\ref{line:value_avg} averages their returns, treating each equilibrium as equally plausible. This preserves behavioral flexibility at execution: a vehicle selects among the retained Nash actions according to the opponent's opinion, rather than being locked into a single equilibrium the opponent may not share. Second, the algorithm does not distinguish the two vehicles by their labels, which avoids selecting an equilibrium during the iteration. Where standard Nash value iteration keeps one value function per player, Algorithm~\ref{alg:vi_ipc} keeps a single $V$, because relabeling the two vehicles leaves the game unchanged (Def.~\ref{def:symmetric}). Accordingly, vehicle~2's value at $s$ is not a second function but the same one evaluated at $\tau(s)$. Because of these two points, the resulting permissive strategy privileges neither vehicle and retains every equilibrium, so the choice among them is left to execution rather than made at synthesis. The algorithm terminates when the value function converges to within a threshold $\theta$ or after $k_{\max}$ sweeps through the state space. 

\begin{algorithm}[h]
\caption{Nash value iteration}\label{alg:vi_ipc}
Initialize $V(s) \leftarrow 0$ for every state $s = (\bar{s}_1, \bar{s}_2)$\;
\For{sweep $k = 1, \ldots, k_{\max}$}{
    $V_{\mathrm{old}} \leftarrow V$\;
    \ForEach{state $s = (\bar{s}_1, \bar{s}_2)$}{
        \lIf{$s \in \mathcal{S}_{\mathrm{fail}}$}{$V(s) \!\leftarrow\! 0$, \textbf{continue}}
        $Q_1, Q_2 \leftarrow \mathbf{0}^{|A_1| \times |A_2|}$\;
        \ForEach{$a = (a_1, a_2) \in A_1 \times A_2$}{
            $s \!\xrightarrow{a}\! s'$\;
            $Q_1(a) \leftarrow R(s, a) + \lambda\, V_{\mathrm{old}}(s')$\;
            $Q_2(a) \leftarrow R(\tau(s), \tau(a)) + \lambda\, V_{\mathrm{old}}(\tau(s'))$\; \label{line:swap}
        }
        $\pi(s) \leftarrow \textsc{Nash}(Q_1, Q_2)$, \ $K(s) \leftarrow |\pi(s)|$\; \label{line:policy_extract}
        $V(s) \leftarrow \dfrac{1}{K(s)} \sum_{(a_1^*,\, a_2^*) \in \pi(s)} Q_1(a_1^*, a_2^*)$\; \label{line:value_avg}
    }
    \lIf{$\max_s |V(s) \!-\! V_{\mathrm{old}}(s)| \!<\! \theta$}{\textbf{break}}
}
\Return{$V$, $\pi$}
\end{algorithm}

\begin{remark}\label{rem:vi_assumptions}
Algorithm~\ref{alg:vi_ipc} returns a well-defined permissive strategy only if the iteration converges and every state outside $\mathcal{S}_{\mathrm{fail}}$ retains at least one pure equilibrium, $K(s)\!\geq\!1$. Neither is guaranteed for general-sum stochastic games~\cite{kearns2000fast, zinkevich2005cyclic}, and characterizing when they hold is outside the scope of this paper. We therefore assume both hold at every state and every iteration in our case study.
\end{remark}

\subsubsection{Opinion-Guided Strategy Design} The permissive joint strategy is the equilibrium set Algorithm~\ref{alg:vi_ipc} returns at each state,
\begin{equation}\label{eq:pi_case3}
\pi(s) \;=\; \bigl\{(a_k^{1*}, a_k^{2*})\bigr\}_{k=1}^{K(s)} \;\subseteq\; A_1\!\times\!A_2,
\end{equation}
where $K(s)\!=\!|\pi(s)|$ counts the equilibria at $s$. Each vehicle's rational action set is then the projection of $\pi(s)$ onto its own action axis, $A_i^\rat(s)\!=\!\{a_k^{i*}\}_{k=1}^{K(s)}$. Given $\pi(s)$, the pair $(\O, \Gamma)$ can be designed to satisfy the conditions of Corollary~\ref{cor:op_type_agnostic} in each state. On the $3\!\times\!3$ grid, every feasible state typically retains one or two equilibria. Where there is a single NE, each vehicle's rational set is a singleton, and no opinion guidance is needed. Where there are two, $\O$ partitions each vehicle's action set into two cells, one for each equilibrium, and $\Gamma$ couples the two opinions so that the committed pair forms one of the retained equilibria. The resulting opinion-guided strategy then selects between the two equilibria according to the inferred opinion of the opponent. Finally, the attention parameter is held constant above the bifurcation threshold~\eqref{eq:rho_set}, $\rho_1\!=\!\rho_2\!>\!\rho^\star$.

\subsubsection{Communication-Free Implementation} To achieve the communication-free realization, we instantiate the tracking controller and the opinion estimator in the physical layer. Specifically, each vehicle's dynamics are represented by a kinematic bicycle model, and the tracking controller $\mu_i$ is realized as a model predictive controller~(MPC) that tracks the reference state associated with the committed action $a_i$; the model and the MPC formulation are given in Appendix~\ref{app:mpc}. The opinion estimator $\Psi_{i}(x(t), h(t))$ is designed around a simple idea: check which of vehicle $j$'s rational actions best matches its observed velocity. The full form is given in Appendix~\ref{app:estimator}.

\begin{figure}[htb]
    \centering
    \subfigure[Resolution A.]{\label{fig:hutong_a}
        \includegraphics[width=0.46\linewidth]{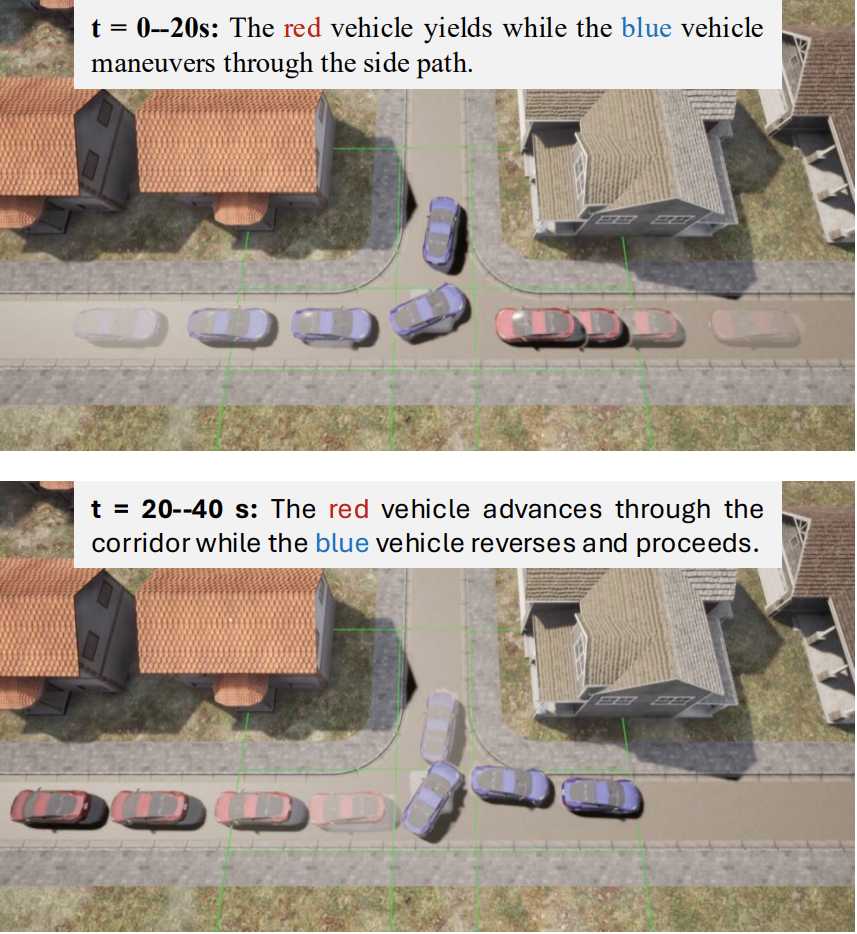}
    }
    \subfigure[Resolution B.]{\label{fig:hutong_b}
        \includegraphics[width=0.46\linewidth]{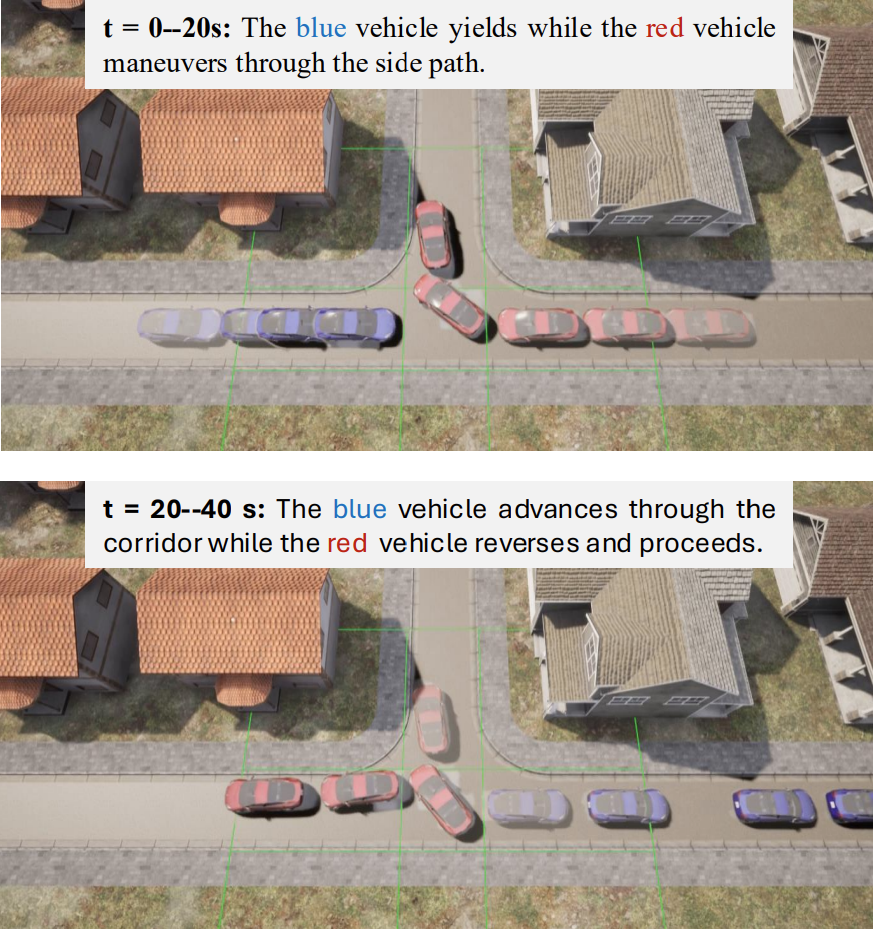}
    }
    \caption{Two symmetric coordination outcomes produced by the proposed framework. Despite sharing an identical strategy, the vehicles spontaneously assume distinct roles via opinion-dynamics-driven symmetry breaking. The green lines mark the boundaries of the grid abstraction.}
    \label{fig:hutong_outcomes}
\end{figure}


\subsubsection{Results}\label{sec:case3_results} Fig.~\ref{fig:hutong_outcomes} shows the two symmetric resolutions that the framework produces across runs of the same encounter, one for each of the game's two equilibria. In Resolution~A (Fig.~\ref{fig:hutong_a}), the blue vehicle gives way, clearing the corridor by pulling into the side path; the red vehicle drives through, and the blue vehicle then reverses out and continues toward its goal. In Resolution~B (Fig.~\ref{fig:hutong_b}), the roles are exactly reversed. In both runs, the opinion states start near neutrality and bifurcate to opposite signs, committing the two vehicles to one of the two retained equilibria; which one is selected is decided by the real-time evolution of the opinions. No communication is involved: each vehicle infers the other's opinion from its observed motion, so the agreement on roles emerges through shared observation alone. These results demonstrate the effectiveness of the opinion-guided strategy in the two-agent game and of its communication-free realization.

Furthermore, we evaluate coordination across all pairings of four strategy types in the scenario of Fig.~\ref{fig:hutong_outcomes}. The three baselines select the equilibrium without regard to the encountered opponent: \emph{NE-A} and \emph{NE-B} commit a vehicle to its role in the corresponding resolution of that figure, and \emph{rule-based} applies a priority rule, giving right of way to the vehicle closer to the conflict region. The proposed \emph{opinion-guided} strategy selects during the interaction, with each vehicle's prior bias $b_i$ sampled independently. Table~\ref{tab:results} reports the success rate of every pairing of these strategy types, over $100$ independent runs each, with vehicle~1 running the row strategy and vehicle~2 the column strategy. As expected, two vehicles committed to the same equilibrium achieve the desired coordination in every run ($100\%$), whereas an NE-A vehicle meeting an NE-B vehicle clashes in every run ($0\%$). The rule-based strategy selects its equilibrium adaptively from the real-time situation, but each vehicle evaluates a rigid criterion on its own observations. Against a fixed-equilibrium opponent, that criterion takes no account of what the opponent has committed to, so the two may settle on different equilibria and the success rate falls to $48$--$65\%$. Against itself the rule does better ($76\%$), since both vehicles evaluate the same condition and usually reach the same equilibrium; the failures that remain occur when the two sit close to that condition's decision boundary, where noise and delay can pull their verdicts apart. The rule-based strategy is thus not robust to the opponent it meets.
The opinion-guided strategy instead adapts to whichever opponent it meets, and the bifurcation is what makes that adaptation robust: the neutral opinion state is unstable, and the only stable ones are the two that assign the vehicles complementary roles, so the opinions are driven into one of them and stay there. A transient misreading of the opponent perturbs that trajectory without changing where it settles, which is exactly what a rigid criterion cannot do. The strategy succeeds in all $100$ runs of every pairing. This result demonstrates that the opinion-guided strategy is robust to the opponent's strategy type as long as its behavior is rational (Corollary~\ref{cor:op_type_agnostic}).

\begin{table}[h]
\centering
\caption{Success rate (\%) between different strategy types, over $100$ independent runs per pairing.}
\label{tab:results}
\begin{tabular}{lcccc}
\toprule
 & \multicolumn{4}{c}{Vehicle~2} \\
\cmidrule(lr){2-5}
Vehicle~1 & NE-A & NE-B & Rule-based & Opinion-guided \\
\midrule
NE-A           & 100 & 0 & 52 & 100 \\
NE-B           & 0 & 100 & 65 & 100 \\
Rule-based     & 48 & 51 & 76 & 100 \\
Opinion-guided & 100 & 100 & 100 & 100 \\
\bottomrule
\end{tabular}
\end{table}

\begin{figure}
	\centering
	\subfigure[Overall trajectory]{\label{fig:hutong_overall}
		\includegraphics[width=0.92\linewidth]{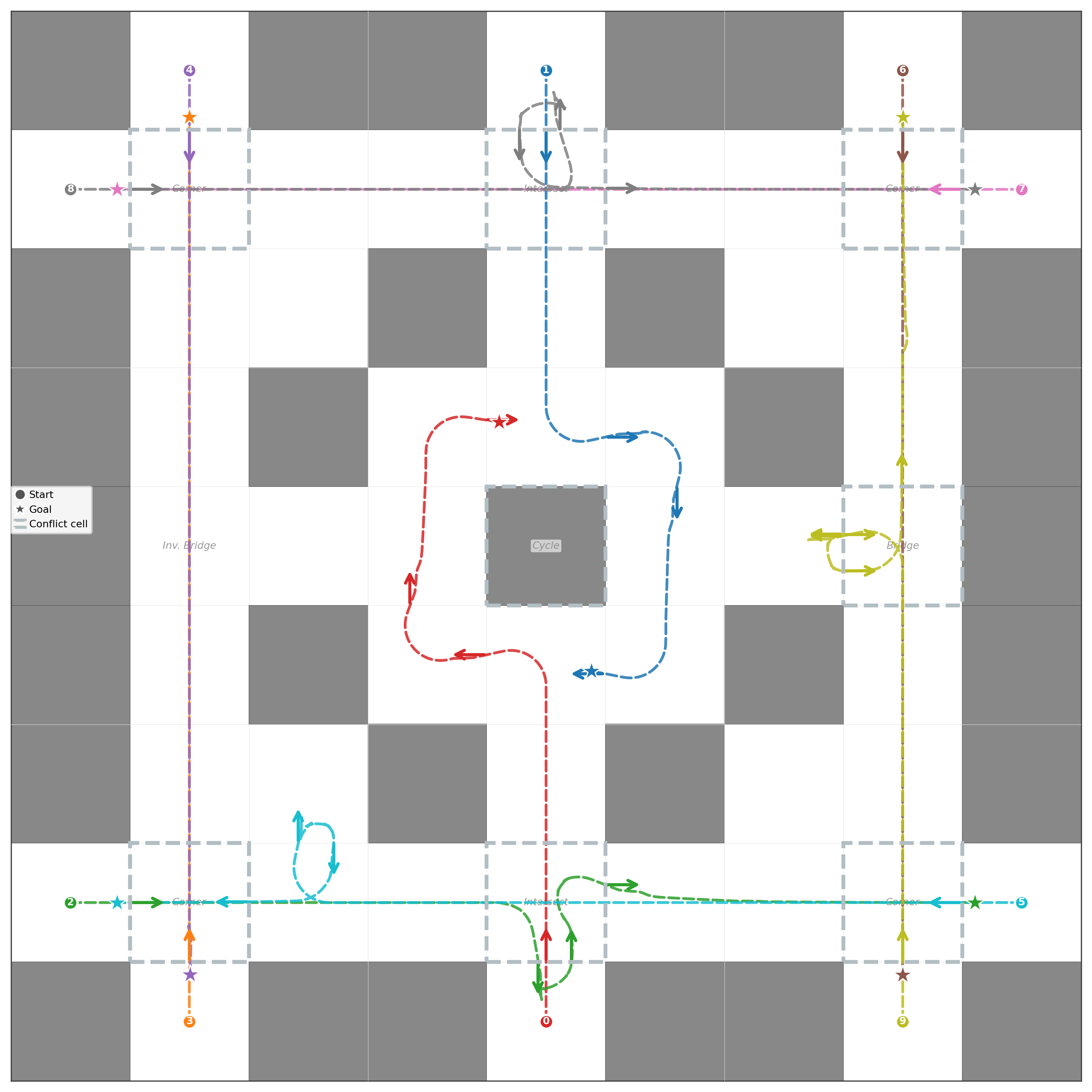}
	}\\
	\subfigure[Submap 1]{\label{fig:hutong_sub1}
		\includegraphics[width=0.44\linewidth]{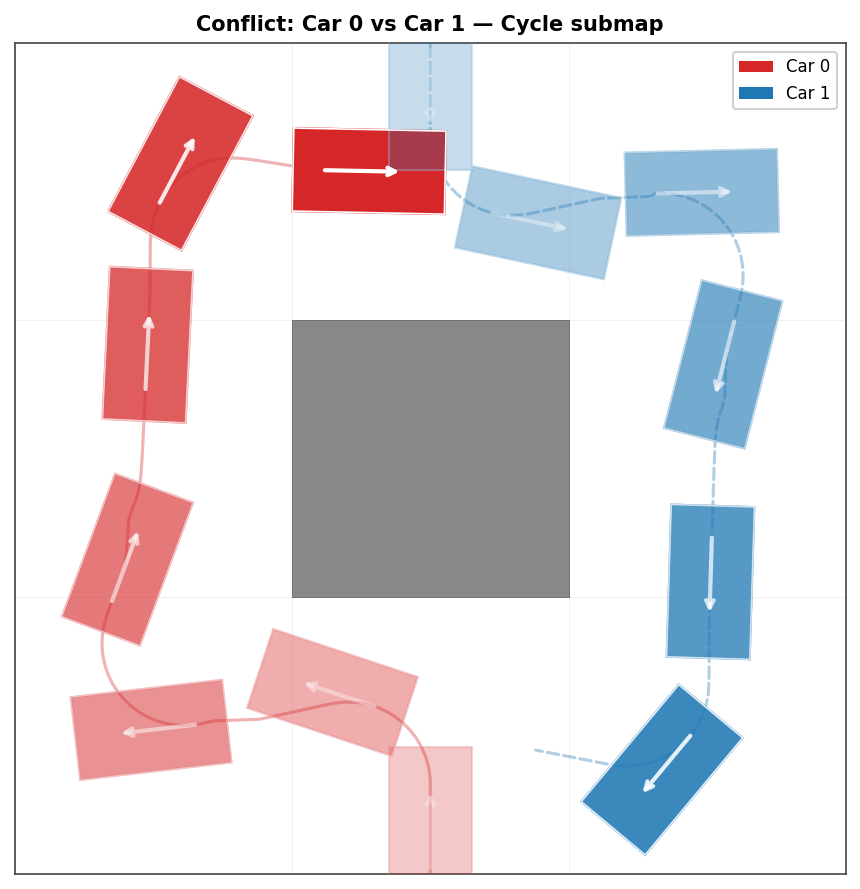}
	}
	\subfigure[Submap 2]{\label{fig:hutong_sub2}
		\includegraphics[width=0.44\linewidth]{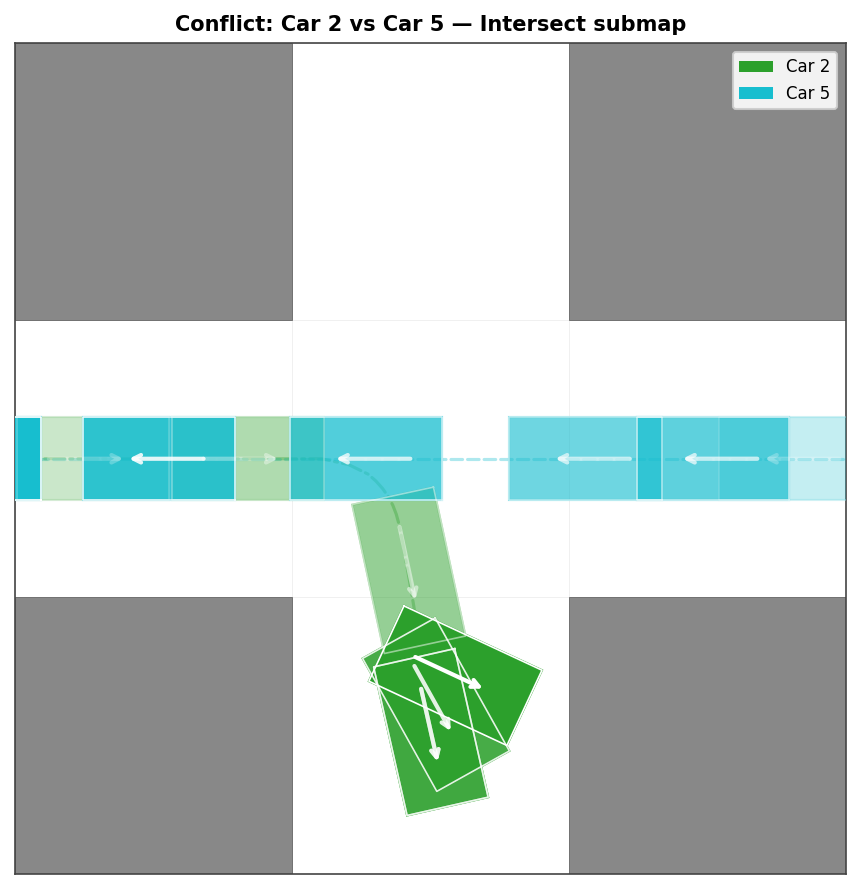}
	}
	\caption{Large-scale validation in a constrained environment: (a) overall trajectory of the fleet; (b)--(c) zoomed-in views of two representative interaction segments.}
	\label{fig:hutong_large}
\end{figure}

Finally, we scale the scenario from a single encounter to a fleet of ten vehicles navigating a constrained environment that resembles a dense residential area or parking lot, each with its own start and goal. Fig.~\ref{fig:hutong_overall} shows the executed trajectories: drivable grid states are white, and grid states blocked by buildings are grey. Each vehicle is assumed to meet at most one other at a time, within the $3\!\times\!3$ neighborhood map around the contested passage, and the permissive strategy for each such neighborhood is generated by Algorithm~\ref{alg:vi_ipc}. The opinion-guided strategy runs on top of it without communication. These neighborhoods are not alike: a corner, an intersection, a one-lane bridge, and the cycle around a central block each induce a different local game, with its own permissive set and its own number of equilibria. The same procedure synthesizes the opinion-guided strategy in every case, so the framework is not tuned to one geometry. From the trajectories, we can see that every vehicle reaches its goal without collision or deadlock, and every passage that admits more than one resolution is settled during the encounter itself. Two of them are enlarged. In Fig.~\ref{fig:hutong_sub1}, vehicles~0 and~1 round the central block and contend for the same grid state; their opinions bifurcate to opposite signs, so one claims the grid state and drives on while the other lets it pass before continuing. In Fig.~\ref{fig:hutong_sub2}, vehicles~2 and~5 meet at an intersection, where vehicle~5 claims the passage and drives straight through while vehicle~2 concedes and proceeds once the intersection clears. Throughout the run, the layered opinion-guided strategy decides on its own when coordination is needed, which equilibrium the two vehicles settle on, and which role each of them takes. The effectiveness across these different geometries demonstrates the generality of the approach, which resolves various encounter scenarios without fixed priorities, hand-designed rules, or communication.

\subsection{Discussion}

This paper is motivated by a critical observation: in practical interaction between agents, there is rarely a single admissible joint behavior, and when each agent selects the one it prefers, the two choices may be incompatible and coordination fails. This work departs from the common premise of a known or communicating opponent and adopts a weaker assumption: the opponent is rational, in the sense that at every moment it is willing to take an action that can be part of an admissible joint behavior. On this basis, we formalized the decentralized coordination problem between two agents through a series of definitions, among them the permissive joint strategy and rationality. We then introduced the requirements that a strategy must satisfy when agents carrying different strategies meet at random: robust coordination, which asks the strategy to remain compatible with every rational opponent that respects the shared permissive set, and homogeneous implementability, which asks it to remain effective when both agents run that same strategy. To meet both requirements, we presented one key insight for designing such a strategy: keep every admissible joint behavior open, preserving the flexibility to adapt to whichever opponent is encountered, and decide the specific one at runtime from how that opponent behaves. Building on this insight, we proposed the opinion-guided strategy and its layered realization. As the case studies show, the framework resolves miscoordination among agents with differing preferences and breaks symmetry among structurally identical agents, without fixed rules, priorities, or communication.  

In addition, several problems remain open for future work: (i)~The present analysis addresses the interaction of two agents over a binary decision, and extending the opinion-guided strategy to more agents and more options is not straightforward. For more than two agents, the bifurcation of the opinion dynamics no longer guarantees that the agents will settle on complementary roles. Although a vector opinion can in principle encode more than two options, Assumption~\ref{asm:linear_pairing} couples each of its components to at most one component of the other agent's opinion, so the components evolve as independent binary decisions; a single choice among more than two mutually related options, which requires dependence among the components, is therefore out of its reach. (ii)~The geometry of the permissive set varies from state to state, and with it what the interaction demands of the partition and the pairing map. So far the two have been designed either from a simple principle, as in Cases~1 and~3, or by hand, as in Case~2. How to design the opinion-pairing map $\Gamma$ and the partition $\O$ so that they satisfy the conditions of Corollary~\ref{cor:op_type_agnostic} for an arbitrary permissive geometry, and in a general multi-agent setting, remains open.
Resolving these two points would generalize the opinion-guided strategy to interactions among more agents and to more complex scenarios.


\section{Conclusions}

This paper studied decentralized coordination between agents that meet at random, without communication and without knowledge of each other's strategy, so that the strategies they carry may prove incompatible. To make the problem precise, we introduced a formal formulation built on two notions: the permissive joint strategy, which collects all the admissible joint behaviors, and rationality, which constrains an agent's individual behavior. On that basis, we stated the requirements that a strategy must meet in decentralized coordination: robust coordination, which asks it to remain compatible with every rational opponent that respects the shared permissive set, and homogeneous implementability, which asks it to remain effective when both agents run that same strategy. 

To satisfy both requirements, we proposed an opinion-guided strategy that couples the two agents through an auxiliary opinion state, realized in a layered framework: the strategy layer keeps every admissible joint behavior available and commits to an action once the other agent's opinion is inferred, while coupled opinion dynamics drive the two agents to a matched pair of opinions. A layer-by-layer analysis established the conditions under which two individually synthesized strategies are compatible, and thereby guarantees that randomly encountered agents realize an acceptable joint behavior. Three case studies of increasing decision complexity, ramp merging, dual-arm manipulation, and navigation in a constrained environment, showed the framework resolving conflicts and bringing the agents to a common admissible joint behavior, without relying on communication or predesigned coordination. Notably, the last of these corresponds to a typical general-sum game admitting multiple Nash equilibria: there the opinion dynamics guided the two agents to a common equilibrium from the opponent's real-time behavior alone, and kept them compatible whether that opponent ran a different strategy or the identical one, provided only that it was rational.

\appendices
\section{Implementation Details of Case-3}
\subsection{Reward Function}\label{app:reward}
The joint state $s\!=\!(\bar{s}_i, \bar{s}_j)$ collects the augmented states $\bar{s}_i\!=\!(s_i, g_i)$ and $\bar{s}_j\!=\!(s_j, g_j)$, each pairing a vehicle's grid state with its goal, and the joint action is $a\!=\!(a_i, a_j)$; the vehicle listed first is taken as the ego. Let $s_i'$ denote the grid state that vehicle $i$ reaches from $s_i$ under action $a_i$, and likewise $s_j'$. The shared reward of Section~\ref{sec:case_study} is
\begin{equation}\label{eq:reward_case3}
\begin{aligned}
R(s, a) = & -1 + 20\cdot\mathbf{1}[s_i'\!=\!g_i]  - 50\cdot\mathbf{1}[s_i'\!\notin\!S] \\
& - 100\bigl(\mathbf{1}[s_i'\!\in\!S_{\mathrm{obs}}] + \mathbf{1}[s_i'\!=\!s_j'] \\ & +  \mathbf{1}[s_j\!=\!s_i' \wedge s_i\!=\!s_j']\bigr).
\end{aligned}
\end{equation}
The constant step cost encourages progress over delay, reaching the goal pays $20$, and leaving the grid costs $-50$. The three $-100$ charges cover the collision events: entering an obstacle grid state, both vehicles entering the same grid state, and the head-on exchange in which each enters the grid state the other vacates. Since $R$ reads only the ego goal $g_i$, it is the swap $\tau$ that makes the two vehicles' payoffs differ: the same function, evaluated at $s$ for one vehicle and at $\tau(s)$ for the other, scores the interaction in terms of their respective goals.

\subsection{Tracking Controller of Case-3}\label{app:mpc}
The full physical state of vehicle $i$ is $\xi_i = (x_i, \varphi_i, v_i)$, which extends the planar position $x_i\!\in\!\mathbb{R}^2$ with the heading $\varphi_i$ and the speed $v_i$. Its motion follows the kinematic bicycle model
\begin{equation}\label{eq:bicycle}
\dot{x}_i = v_i \begin{bmatrix} \cos\varphi_i \\ \sin\varphi_i \end{bmatrix}, \quad
\dot{\varphi}_i = \frac{v_i}{L}\tan\delta_i, \quad
\dot{v}_i = \eta_i,
\end{equation}
where $L$ is the wheelbase and the input $u_i = (\eta_i, \delta_i)$ collects the longitudinal acceleration and the steering angle. 
For digital implementation, the dynamics~\eqref{eq:bicycle} are discretized into $f_d$ at the sampling time $\Delta t_{\mathrm{c}}$ of the physical layer, with $\Delta t_{\mathrm{c}} \!\ll\! t_{k+1} \!-\! t_k$. Given the action $a_i$, the tracking controller $\mu_i$ solves a finite-horizon problem over $N$ steps of $f_d$ to track the reference state $\xi_i^r$ associated with $a_i$:
\begin{equation}\label{eq:mpc_case3}
\begin{aligned}
\min_{u_{i,0:N-1}} \ & \sum_{n=0}^{N-1} \Bigl( \|\xi_{i,n} - \xi_i^r\|_{W}^2 + \|u_{i,n}\|_{W_u}^2 \Bigr) + \|\xi_{i,N} - \xi_i^r\|_{W_f}^2 \\
\text{s.t.} \ & \xi_{i,n+1} = f_d(\xi_{i,n}, u_{i,n}), \quad \xi_{i,0} = \xi_i(t), \\
& u_{i,n} \in \mathcal{U}, \quad n = 0, \ldots, N-1,
\end{aligned}
\end{equation}
where $W, W_f \!\succeq\! 0$ and $W_u \!\succ\! 0$ are weighting matrices, and $\mathcal{U}$ bounds the acceleration and the steering angle. The reference $\xi_i^r$ encodes the committed action $a_i$. The first input of the minimizing sequence is applied and the problem is re-solved at the next sampling instant, in a receding-horizon fashion.

\subsection{Opinion Estimator of Case-3}\label{app:estimator}
Each maneuver in the action set $A_j\!=\!\{\texttt{forward},\, \texttt{back},\, \texttt{left},\, \texttt{right},\, \texttt{wait}\}$ induces a nominal motion at the current state; let $v_{j}^{(a)}$ denote the reference velocity of maneuver $a\!\in\!A_j$ for vehicle $j$, and $v_{j}(t)$ the observed velocity of vehicle $j$ at time $t$. The similarity of the observed velocity to maneuver $a$ is quantified on a common $[0,1]$ scale as
\begin{equation}\label{eq:similarity_case} d_{j}^{a}(t) = \begin{cases} \max\Bigl\{0,\ \dfrac{\langle v_{j}(t),\, v_{j}^{(a)}\rangle}{\|v_{j}(t)\|\,\|v_{j}^{(a)}\| + \epsilon}\Bigr\}, & a \neq \texttt{wait},\\[8pt] \max\Bigl\{0,\ 1 - \dfrac{\|v_{j}(t)\|}{\bar{v}}\Bigr\}, & a = \texttt{wait}. \end{cases} \end{equation}
Here $\bar{v}$ is the nominal speed. Thus, a moving maneuver is scored by direction alone and $\texttt{wait}$ by proximity to rest, placing the two on the same scale. Each opinion state $\tilde{z}$ is then scored by the best-matching rational action in $A_j^{(\tilde{z})}$: $d_{j}^{(\tilde{z})}(t) = \max_{a \in A_j^\rat(s) \cap A_j^{(\tilde{z})}} d_{j}^{a}(t)$. Thus, the opinion estimate is
\begin{equation}\label{eq:estimator_case} \hat{z}_{j}(t) = k_z\, \frac{d_{j}^{(+1)}(t) - d_{j}^{(-1)}(t)}{\sqrt{\bigl(d_{j}^{(+1)}(t) + d_{j}^{(-1)}(t)\bigr)^2 + \epsilon}}, \end{equation}
where $k_z > 0$ is a scaling factor and $\epsilon > 0$ ensures numerical stability. Both cell scores are non-negative, so $|\hat{z}_{j}(t)|\!<\!k_z$, and the sign of $\hat{z}_{j}(t)$ records which cell the observed motion falls in.

\bibliographystyle{ieeetr}
\bibliography{MyBib}

\end{document}